\documentclass[11pt,a4paper]{article}
\usepackage{ifxetex}
\ifxetex
  \usepackage[no-math]{fontspec}
\else
\fi
\usepackage{amsmath}
\usepackage{amsfonts}
\usepackage{amssymb}
\usepackage{amsthm}
\usepackage{fullpage}
\usepackage{microtype}
\usepackage[libertine]{newtxmath}
\usepackage[tt=false]{libertine} 
\usepackage{caption}
\usepackage{bbm}
\usepackage{hyperref, color}
\hypersetup{colorlinks=true,citecolor=blue, linkcolor=blue, urlcolor=blue}
\usepackage[linesnumbered,boxed,ruled,vlined]{algorithm2e}
\usepackage{bm}
\usepackage{bbm}
\usepackage[numbers]{natbib}
\usepackage{xcolor}
\usepackage{enumerate} 
\usepackage{enumitem}
\usepackage{tabularx}
\usepackage{array}
\usepackage{cleveref}
\setitemize{noitemsep, topsep=0pt}
\setenumerate{noitemsep, topsep=0pt}
\newcolumntype{L}[1]{>{\raggedright\arraybackslash}p{#1}}
\newcolumntype{C}[1]{>{\centering\arraybackslash}m{#1}}
\newcolumntype{R}[1]{>{\raggedleft\arraybackslash}p{#1}}

\usepackage{makecell}
\usepackage{footnote}
\usepackage{float}
\makesavenoteenv{tabular}

\usepackage{mleftright}
\usepackage{longtable}

\newcommand{\dist}{\text{\rm dist}}
\newcommand{\abs}[1]{\left\vert#1\right\vert}

\newcommand{\upt}{\mathsf{time_{ud}}}

\newtheorem{theorem}{Theorem}

\newtheorem{lemma}[theorem]{Lemma}

\newtheorem{observation}[theorem]{Observation}
\newtheorem{proposition}[theorem]{Proposition}

\newtheorem{corollary}[theorem]{Corollary}

\newtheorem{definition}[theorem]{Definition}
\newtheorem*{remark}{Remark}

\theoremstyle{plain}

\crefname{theorem}{Theorem}{Theorems}
\crefname{observation}{Observation}{Observations}
\crefname{claim}{Claim}{Claims}
\crefname{condition}{Condition}{Conditions}
\crefname{algorithm}{Algorithm}{Algorithms}
\crefname{property}{Property}{Properties}
\crefname{example}{Example}{Examples}
\crefname{fact}{Fact}{Facts}
\crefname{lemma}{Lemma}{Lemmas}
\crefname{corollary}{Corollary}{Corollaries}
\crefname{definition}{Definition}{Definitions}
\crefname{remark}{Remark}{Remarks}
\crefname{proposition}{Proposition}{Propositions}
\crefname{equation}{equation}{equations}
\crefname{enumi}{Case}{Case}
\creflabelformat{enumi}{(#2#1#3)}

\usepackage[textsize=tiny]{todonotes}

\newcommand{\tp}[1]{\left(#1\right)}
\newcommand{\sqtp}[1]{\left[#1\right]}

\newcommand{\sample}{\mathsf{Sample}}
\newcommand{\Rejsample}{\mathsf{RejectionSampling}}
\newcommand{\DTV}[2]{d_{\mathrm{TV}}({#1},{#2})}
\newcommand{\edge}[1]{\mathsf{edge}\tp{#1}}

\renewcommand{\Pr}{\mathop{\mathrm{Pr}}\nolimits}

\newcommand{\mixingtime}[1]{\ensuremath{t_{\textnormal{mix}}(#1)}}

\newcommand{\defeq}{:=}
\newcommand{\ctp}[1]{\left\lceil{#1}\right\rceil}
\newcommand{\ftp}[1]{\left\lfloor{#1}\right\rfloor}
\newcommand{\ext}[1]{{#1}^{\mathrm{ext}}}

\def\*#1{\bm{#1}}
\def\-#1{\mathsf{#1}}
\def\+#1{\mathcal{#1}}
\def\=#1{\mathbb{#1}}

\def\vbl{\-{vbl}}

\DeclareMathOperator{\Lin}{Lin}

\SetKwComment{Comment}{/* }{ */}

\newbool{doubleblind}
\setbool{doubleblind}{false}

\ifdoubleblind
{
  \title{Improved bounds for randomly colouring simple hypergraphs
  \\(full version)
}
\else
{
  \title{Improved bounds for randomly colouring simple hypergraphs\thanks{This project has received funding from the European Research Council (ERC) under the European Union's Horizon 2020 research and innovation programme (grant agreement No. 947778)}
  }
  \author{
  Weiming Feng\thanks{School of Informatics, University of Edinburgh, Informatics Forum, Edinburgh, EH8 9AB, United Kingdom. \textnormal{E-mails: \url{wfeng@ed.ac.uk}, \url{hguo@inf.ed.ac.uk}, \url{jiaheng.wang@ed.ac.uk}}}
  \and
  Heng Guo\footnotemark[2]
  \and
  Jiaheng Wang\footnotemark[2]
  }
}
\fi

\date{}

\begin{document}

\maketitle

\begin{abstract}
We study the problem of sampling almost uniform proper $q$-colourings in $k$-uniform simple hypergraphs with maximum degree $\Delta$.
For any $\delta > 0$, if $k \geq\frac{20(1+\delta)}{\delta}$ and $q \geq 100\Delta^{\frac{2+\delta}{k-4/\delta-4}}$, 
the running time of our algorithm is $\tilde{O}(\mathrm{poly}(\Delta k)\cdot n^{1.01})$, where $n$ is the number of vertices.  
Our result requires fewer colours than previous results for general hypergraphs (Jain, Pham, and Voung, 2021; He, Sun, and Wu, 2021),
and does not require $\Omega(\log n)$ colours unlike the work of Frieze and Anastos (2017).
\end{abstract}

\section{Introduction}

The past few years have witnessed a bloom in techniques targeted at approximate counting and sampling problems, among which constraint satisfaction problems (CSPs) are probably the most studied. In fact, many problems can be cast as CSPs, e.g., Boolean satisfiability problems (SATs), proper colourings of graphs and hypergraphs, and independent sets, to name a few. 
In general, even deciding if a CSP instance can be satisfied or not is $\mathbf{NP}$-hard.
However, efficient algorithms become possible when the number of appearances of each variable (usually referred to as the degree) is not too high. 
For these instances, the Lov{\' a}sz Local Lemma \cite{EL75} provides a fundamental criterion to guarantee the existence of a solution. 
Although the original local lemma does not provide an efficient algorithm,
after two decades of effort \cite{Beck91,Alon91,MR98,CS00,Sri08,Mos09}, the celebrated work of Moser and Tardos \cite{MT10} provides an efficient algorithm matching the same conditions as the local lemma. 

Unfortunately, the output distribution of the Moser--Tardos algorithm does not suit the need of approximate counting and sampling. 
This deficiency is fundamental, as the sampling problem can be $\mathbf{NP}$-hard even when the criterion of the local lemma is satisfied and the corresponding searching problem lies in $\mathbf{P}$ \cite{BGGGS19,GGW21}. 
In other words, sampling problems are fundamentally more difficult than searching problems in the local lemma regime.
Part of the difficulty comes from the possibility that the state space can be disconnected from local moves,
but traditional algorithmic tools like Markov chain Monte Carlo rely on the connectivity.
This barrier has been bypassed recently by some exciting developments \cite{Moi19,GJL19,GLLZ19,JPV20}, and in particular the projected Markov chain approach \cite{FGYZ21,FHY20,JPV21,HSW21}.
For searching problems, the local lemma is known to give a sharp computational transition threshold from $\mathbf{P}$ to $\mathbf{NP}$-hard \cite{MT10,GST16} as the degree increases.
Recent efforts aim to find and establish a similar threshold for sampling problems as well.

One very promising problem to establish such a threshold is (proper) $q$-colourings of hypergraphs,
which is the original setting where the local lemma was developed \cite{EL75}, and has received considerable recent attention. 
A colouring of a hypergraph is \emph{proper} if no hyperedge is monochromatic.
An efficient (perfect) sampler exists when $q\gtrsim \Delta^{3/(k-4)}$ (where $\gtrsim$ or $\lesssim$ hides some constant independent from $q$, $k$, and $\Delta$)
for $k$-uniform hypergraphs with maximum degree $\Delta$ \cite{JPV21,HSW21}, 
while the sampling problem is $\mathbf{NP}$-hard whenever $q\lesssim \Delta^{2/k}$ for even $q$ \cite{GGW21}.
For comparison, the local lemma shows that a proper $q$-colouring exists if $q\gtrsim \Delta^{1/(k-1)}$ 
(see also \cite{wanless2020general} for a recent alternative approach leading to a slightly better constant). 

On the other hand, before the recent wave of local lemma inspired sampling algorithms, 
randomly sampling $q$-colourings in \emph{simple} $k$-uniform hypergraphs\footnote{A hypergraph is \emph{simple} if any two hyperedges intersect in at most one vertex. Simple hypergraphs are also known as linear hypergraphs.} has already been studied \cite{FM11,FA17}.
In particular, Frieze and Anastos \cite{FA17} gave an efficient sampling algorithm 
when the number of colours satisfies $q \geq \max\{C_k \log n, 500k^3\Delta^{\frac{1}{k-1}}\}$, where $n$ is the number of vertices and $C_k$ depends only on $k$.
Their algorithm is the standard Glauber dynamics with a random initial (not necessarily proper) colouring.
The logarithmic lower bound on the number of colours is crucial to their analysis,
as it guarantees that there is a giant connected component in the state space so that connectivity is not an issue.

In this paper, we study the projected Markov chain for sampling $q$-colourings in simple hypergraphs.
Our result improves the bound of \cite{JPV21,HSW21} for general hypergraphs,
and does not require unbounded number of colours, unlike in \cite{FM11,FA17}.
Let $\mu$ denote the uniform distribution over all proper colourings. 
Our main result is stated as follows.

\begin{theorem}\label{theorem-main}
  For any $\delta > 0$, there is a sampling algorithm such that given any $\epsilon \in (0,1)$, a $k$-uniform simple hypergraph $H=(V,E)$ with maximum degree $\Delta$,
  where $k \geq \frac{20(1+\delta)}{\delta}$, and an integer $q \geq 100\Delta^{\frac{2+\delta}{k - 4/\delta - 4}}$,
  it returns a random $q$-colouring that is $\epsilon$-close to $\mu$ in total variation distance in time $\tilde{O}( k^5\Delta^2 n \tp{\frac{n\Delta}{\epsilon}}^{0.01})$, 
  where $n = |V|$ and $\tilde{O}$ hides a $\mathrm{polylog}(n,\Delta,q,1/\epsilon)$ factor. 
\end{theorem}

A few quick remarks are in order.
First of all, the exponent of $n$ in the running time can be made even closer to $1$ if more colours are given. See \Cref{theorem-general-main} for the full technical statement. Secondly, our algorithm can be modified into a perfect sampler by applying the bounding chain method \cite{huber1998exact} based on coupling from the past (CFTP) \cite{PW96}, following the same lines of \cite{HSW21}. 
Moreover, using known reductions from approximate counting to sampling~\cite{jerrum1986random,vstefankovivc2009adaptive,huber2015approximation,kolmogorov18faster} (see \cite{FGYZ21} for simpler arguments specialized to local lemma settings), 
one can efficiently and approximately count the number of proper colourings in simple hypergraphs under the same conditions in \Cref{theorem-main}. 

Our algorithm follows the recent projected Markov chain approach \cite{FGYZ21} with state compression \cite{FHY20}.
Roughly speaking, instead of assigning colours to vertices,
we split $[q]$ into $\sqrt{q}$ buckets of size $\sqrt{q}$ each and assign buckets to vertices.
We run a (systematic scan) Markov chain on these bucket assignments to generate a sample, 
and then conditional on this sample to draw a nearly uniform $q$-colouring.
The benefit of this bucketing is that, under the conditions of \Cref{theorem-main},
conditional on the assignments of all but one vertices,
the assignment of the remaining vertex is close to uniformly at random.
This implies that any atomic event\footnote{An event is \emph{atomic} if each variable it depends on must take one particular value. In discrete spaces, any event can be decomposed into atomic ones.} is exponentially unlikely in the number of distinct vertices it depends on.
In order to show that this approach works, we need to show two things:
1) the projected Markov chain is rapidly mixing;
2) each step of the Markov chain can be efficiently implemented.
For general hypergraphs, the previous $q\gtrsim\Delta^{3/(k-4)}$ bound comes from balancing the conditions so that the two claims are true simultaneously.
However, there is no room left for relaxation on either claim.
This means that, for our improvements in simple hypergraphs, new ingredients are required for both claims.

For rapid mixing, we take the information percolation approach \cite{HSZ19,JPV21,HSW21},
where the main effort is to trace discrepancies through a one-step greedy coupling,
and to show that they are unlikely after a sufficient amount of time.
In simple hypergraphs, an individual discrepancy path through time has more distinct updates of vertices than in the general case, and are thus more unlikely.
This allows us to relax the condition.
Our mixing time analysis is largely inspired by the work of Hermon, Sly, and Zhang \cite{HSZ19},
although we do need to handle some new complicacies, such as hyperedges whose vertices are consecutively updated in the discrepancy path.

For efficient implementation, we use rejection sampling. 
Here we want to sample the colour/bucket of a vertex conditional on the buckets of all other vertices.
We can safely prune hyperedges containing vertices of different buckets.
The remaining connected component containing the update vertex needs to have logarithmic size to guarantee efficiency of our rejection sampling.
The standard approach to bound its size is to do a union bound over certain combinatorial structures with sufficiently many distinct vertices.
Most previous analysis is based on enumerating so-called ``$2$-trees'', a notion first introduced by Alon \cite{Alon91}.
Unfortunately, under the conditions of \Cref{theorem-main}, there are too many ``$2$-trees'' to our need.
Instead, we introduce a new structure called ``$2$-block-trees'' (see \Cref{def:2-block-tree}).
Here each ``block'' is a collection of $\theta$ connected hyperedges,
and these blocks satisfy connectivity properties similar to a $2$-tree.
Since the hypergraph is simple, a block has at least $\theta k-\binom{\theta}{2}$ distinct vertices.
As long as $\theta\ll k$, we have a good lower bound on the number of distinct vertices, 
which in turn implies a good upper bound on the probability of these structures showing up.
To finish off with the union bound, we give a new counting argument for the number of $2$-block-trees,
which is based on finding a good encoding of these structures.

The exponent (roughly $2/k$) of $\Delta$ in \Cref{theorem-main} is unlikely to be tight,
although it appears to be the limit of current techniques.
In fact, we conjecture that the computational transition for sampling $q$-colourings in simple hypergraphs happens around the same threshold of the local lemma (namely, the exponent should be roughly $1/k$).
This conjecture is supported by the hardness result of Galanis, Guo, and Wang \cite{GGW21} for general $q$, and by the algorithm of Frieze and Anastos \cite{FA17} for $q=\Omega(\log n)$.
Note that for a simple $k$-uniform hypergraph with maximum degree $\Delta$, 
Frieze and Mubayi \cite{FM13} showed that the chromatic number $\chi(H)\le C_k\left(\frac{\Delta}{\log\Delta}\right)^{\frac{1}{k-1}}$ where $C_k$ depends only on $k$.
Their bound is asymptotically better than the bound given by the local lemma.
Thus there may still be a gap between the searching threshold and the sampling threshold.

A final remark is that our method would still work as long as the overlap of hyperedges is much smaller than $k$.
The condition on the parameters may deteriorate slightly but would still be better than those for general hypergraphs. 
On the other end of the spectrum, if any two intersecting hyperedges intersect at at least $k/2$ vertices,
the algorithm by Guo, Jerrum, and Liu \cite{GJL19} almost matches the hardness result \cite{GGW21}.
It is an intriguing question how the size of overlaps affects the complexity of these sampling problems,
or whether it is possible to improve sampling algorithms via a better use of the overlap information.

\section{Preliminaries}

In this section we gather some preliminary definitions and results for later use.
We generally use the bold font to denote vectors, matrices, and/or random variables.

\subsection{Graph theory}

Throughout this paper, we use the following notations for a graph $G=(V,E)$:
\begin{itemize}
  \item $G[A]$: the induced subgraph of $G$ on the vertex subset $A \subseteq V$. 
  \item $\dist_G(A,B)$: the distance between two vertex sets $A\subseteq V$ and $B\subseteq V$ on $G$, which is defined by $\dist_G(A,B) \defeq \min_{u \in A,v\in B}\dist_G(u,v)$ and $\dist_G(u,v)$ is the length of the shortest path between $u$ and $v$ in $G$. 
  \item $\Gamma_G^i(A)$: the set of vertices $u$ such that $\dist_G(A,u)=i$. Specifically, when $i=1$, this notation represents the neighbourhood of the given set $A \subseteq V$, and is also denoted by $\Gamma_G(A)$. 
\end{itemize}
We sometimes do not distinguish $u$ and the singleton set $\{u\}$ in sub- or sup-scripts.
For the sake of convenience, we may drop the subscript $G$ when the underlying graph is clear from the context. 

We need some more definitions for later use. 

\begin{definition}[Graph power] \label{def:graph_power}
  Let $G$ be an undirected graph. 
  The $i$-th power of $G$, denoted by $G^i$, is another graph that has the same vertex set as $G$, 
  and $\{u,v\}$ is an edge in $G^i$ iff $1\leq \dist_G(u,v)\leq i$. 
\end{definition}

\begin{definition}[Line graph] \label{def:line_graph}
  Let $H=(V,\+E)$ be a hypergraph. 
  Its line graph $\Lin(H)=(V_L,E_L)$ is given by
  $V_L=\+E$, and $\{e,e'\}\in E_L$ iff $e\cap e'\neq\emptyset$. 
\end{definition}

\subsection{Coupling and Markov chains}

Consider a discrete state space $\Omega$ and two distributions $\mu$ and $\nu$ over it. 
The \emph{total variation distance} between $\mu$ and $\nu$ is defined by
\[
  \DTV{\mu}{\nu}:=\frac{1}{2}\sum_{x\in\Omega}\left\vert\mu(x)-\nu(x)\right\vert. 
\]
A \emph{coupling} between $\mu$ and $\nu$ is a joint distribution $(X,Y)\in\Omega^2$ such that
its marginal distribution over $X$ (resp.~$Y$) is $\mu$ (resp.~$\nu$). 
The next lemma, usually referred to as the \emph{coupling lemma}, bounds the total variation distance between $\mu$ and $\nu$ by any of their couplings. 
\begin{lemma}[Coupling lemma]\label{lem:coupling_lemma}
For any coupling $(X,Y)$ between between $\mu$ and $\nu$, 
\[
  \DTV{\mu}{\nu}\leq\Pr[X\neq Y].
\] 
Moreover, there exists an optimal coupling reaching the equality. 
\end{lemma}

Given a finite state space $\Omega$, a discrete-time \emph{Markov chain} is a sequence $\{X_t\}_{t\geq 0}$
where the probability of each possible state of $X_{t+1}$ only depends on the state of $X_t$. 
The transition of the chain is represented by the \emph{transition matrix} ${\bm P}:\Omega^2\to\={R}_{[0,1]}$, 
where ${\bm P}(i,j)=\Pr[X_{t+1}=j\mid X_t=i]$. 
When the state space $\Omega$ is clear from context, we simply denote the chain by its transition matrix. 
A Markov chain ${\bm P}$ is:
\begin{itemize}
  \item \emph{irreducible}, if for any $X,Y\in\Omega$, there exists $t>0$ such that ${\bm P}^t(X,Y)>0$;
  \item \emph{aperiodic}, if for all $X\in \Omega$, it holds that $\gcd\{t\mid {\bm P}^{t}(X,X)>0\}=1$; and
  \item \emph{reversible} with respect to a distribution $\pi$, if
  \[
    \pi(X){\bm P}(X,Y)=\pi(Y){\bm P}(Y,X) \qquad \forall X,Y\in\Omega.
  \]
  This equation is usually known as the \emph{detailed balance condition}. 
\end{itemize}
A distribution $\pi$ is \emph{stationary} for ${\bm P}$, if $\pi{\bm P}=\pi$ (regarding $\pi$ as a row vector). 
The detailed balance condition actually implies that the corresponding distribution is stationary. 
Furthermore, if a Markov chain is both irreducible and aperiodic, then it converges to a unique stationary distribution $\pi$. 
The speed of convergence towards $\pi$ is characterised by its \emph{mixing time}, defined by
\[
  \mixingtime{{\bm P},\epsilon}:=\min\left\{t \mid \max_{X\in\Omega}\DTV{{\bm P}^t(X,\cdot)}{\pi}<\epsilon\right\}.
\]

The joint process $(X_t,Y_t)_{t \geq 0}$ is a \emph{coupling of Markov chain} $\*P$ if $(X_t)_{t \geq 0}$ and $(Y_t)_{t\geq 0}$ individually follow the transition rule of $\*P$, and if $X_i = Y_i$ then $X_{j} = Y_j$ for all $j \geq i$. 
By the coupling lemma, for any coupling $(X_t,Y_t)_{t \geq 0}$ of $\*P$, it holds that
\begin{align*}
\DTV{P^t(X_0,\cdot)}{P^t(Y_0,\cdot)} \leq \Pr[X_t \neq Y_t]. 
\end{align*}
Hence, the mixing time of $\*P$ can be bounded by
\begin{align}\label{eqn:coupling-mixing}
\mixingtime{{\bm P},\epsilon} \leq \max_{X_0,Y_0 \in \Omega}	 \min\left\{t \mid \Pr[X_t \neq Y_t] \leq \epsilon \right\}.
\end{align}

\subsection{Lov{\'a}sz Local Lemma}

Let $\+{R}=\{R_1,\cdots,R_n\}$ be a set of mutually independent random variables. 
Given an event $A$, denote the set of variables that determines $A$ by $\vbl(A)\subseteq\+R$. 
Let $\+B=\{B_1,\cdots,B_n\}$ be a collection of ``bad'' events. 
For any event $A$ (not necessarily in $\+B$), let $\Gamma(A):=\{B\in\+B\mid B\neq A,\;\vbl(B)\cap\vbl(A)\neq\emptyset\}$. 
We will use the following version of \emph{Lov{\'a}sz Local Lemma} from \cite{HSS11}.
\begin{theorem}[\text{\cite{EL75,HSS11}}] \label{thm:lll}
  If there exists a function $x:\+B\to(0,1)$ such that for any bad event $B\in\+B$, 
  \begin{equation}  \label{equ:lll_cond}
    \Pr[B]\leq x(B)\prod_{B'\in\Gamma(B)}(1-x(B')),
  \end{equation}
  then it holds that
  \[
    \Pr\left[\bigwedge\limits_{B\in\+B}\bar{B}\right]\geq\prod_{B\in\+B}(1-x(B))>0.
  \]
  Moreover, for any event $A$,
  \begin{equation}  \label{equ:lll}
    \Pr\left[A\mid \bigwedge\limits_{B\in\+B}\bar{B}\right]\leq\Pr[A]\prod_{B\in\Gamma(A)}(1-x(B))^{-1}. 
  \end{equation}
\end{theorem}

\subsection{List hypergraph colouring and local uniformity}

In our algorithm and analysis, we consider the general \emph{list hypergraph colouring} problem. 
Let $H=(V,\+E)$ be a $k$-uniform hypergraph with maximum degree $\Delta$. 
Let $(Q_v)_{v \in V}$ be a set of colour lists.
We say $\*X \in \otimes_{v \in V} Q_v$ is a proper list colouring if no hyperedge in $H$ is monochromatic with respect to $\*X$.
Let $\mu$ denote the uniform distribution of all proper list hypergraph colourings.
The following local uniformity property holds for the distribution $\mu$.
Its proof follows from the argument in~\cite{GLLZ19}. We include it here for completeness.

\begin{lemma}[local uniformity~\text{\cite{GLLZ19}}]\label{lemma-local-uniform-colour}
Let $q_0 = \min_{v \in V} \abs{Q_v}$ and $q_1 = \max_{v \in V} \abs{Q_v}$.
For any $r \geq k \geq 2$, if $q_0^k \geq \mathrm{e}q_1 r \Delta $, the for any $v\in V$ and $c \in Q_v$,
\begin{align*}
\frac{1}{|Q_v|}\exp\tp{-\frac{2}{r}}	 \leq \mu_v(c) \leq \frac{1}{|Q_v|}\exp\tp{\frac{2}{r}},	
\end{align*}
where $\mu_v$ is the marginal distribution on $v$ induced by $\mu$.
\end{lemma}
\begin{proof}
Let $\+D$ denote the product distribution where each $v \in  V$ samples a colour in $Q_v$ uniformly at random. 
For each $e \in \+E$, let $B_e$ be the bad event that $e$ is monochromatic.	
Let $x(e) = \frac{1}{r\Delta}$ for all $e \in \+E$. 
Note that $r \geq k$.
We have
\begin{align*}
\Pr_{\+D}\sqtp{B_e} \leq \frac{q_1}{q_0^k} \leq \frac{1}{\mathrm{e}r\Delta} \leq \frac{1}{r\Delta}\tp{1 - \frac{1}{r\Delta}}^{k(\Delta -1)} \leq x(B_e)\prod_{B \in \Gamma(B_e)}(1-x(B)).	
\end{align*}
By \Cref{thm:lll}, it holds that 
\begin{align*}
\mu_v(c) \leq \frac{1}{\abs{Q_v}}\tp{1 - \frac{1}{r\Delta}}^{-\Delta} \leq \frac{1}{\abs{Q_v}}\exp\tp{\frac{2}{r}}.	
\end{align*}

For the lower bound, consider each hyperedge $e$ such that $v \in e$. 
Let $\mathsf{Block}_e$ be the event that all vertices in $e$ except $v$ have the colour $c$.
%
%
If none of $\mathsf{Block}_e$ occurs, then $v$ has colour $c$ with probability at least $1/{\abs{Q_v}}$.
By \Cref{thm:lll}, we have
\begin{align*}
\mu_v(c) \geq \frac{1}{\abs{Q_v}} \Pr_{\mu}\sqtp{ \bigwedge_{e \ni v} \overline{\mathsf{Block}_e} }	\geq \frac{1}{\abs{Q_v}}\tp{1 - \sum_{e \ni v}\Pr_{\mu}\sqtp{\mathsf{Block}_e}}.
\end{align*}
Note that $\Pr_{\+D}\sqtp{\mathsf{Block}_e} \leq q_0^{-k+1}$ and $|\Gamma(\mathsf{Block}_e)|\leq k(\Delta-1) + 1$. We have
\begin{align*}
\Pr_{\mu}\sqtp{\mathsf{Block}_e} \leq q_0^{-k+1}\tp{1 - \frac{1}{r\Delta}}^{-k(\Delta-1)-1} \leq q_0^{-k+1} \mathrm{e} \leq \frac{1}{r\Delta},
\end{align*}
where the last inequality holds because $q_0^{-k+1} \mathrm{e} \leq q_0^{-k} q_1 \mathrm{e} \leq \frac{1}{r\Delta}$,
which implies
\begin{align*}
\mu_v(c) \geq \frac{1}{\abs{Q_v}}\tp{1 - \sum_{e \ni v}\Pr_{\mu}\sqtp{\mathsf{Block}_e}} \geq \frac{1}{\abs{Q_v}}\tp{1 - \frac{1}{r}} \geq 	\frac{1}{\abs{Q_v}}\exp\tp{-\frac{2}{r}}. &\qedhere
\end{align*}

\end{proof}

\section{Algorithm}
Let $H = (V,\+E)$ be a $k$-uniform hypergraph and $[q]$ a set of colours.
Let $\mu$ denote the uniform distribution of proper hypergraph colourings. 
%
%
Our algorithm is a variant of the projected dynamics from~\cite{FGYZ21}, using a particular projection scheme from~\cite{FHY20}.
We first introduce some basic definitions and notations,  and then describe the sampling algorithm.

\subsection{Projection scheme, projected distribution and conditional distribution}

Our sampling algorithm is based on the following \emph{projection scheme} introduced in~\cite{FHY20}.

\begin{definition}[projection scheme \text{\cite{FHY20}}]\label{definition-projection}
Let $1\leq s \leq q$ be an integer. A (balanced) projection scheme with image size $s$ is a function $h: [q] \to [s]$ such that for any $j \in [s]$, $\abs{h^{-1}(j)} = \lfloor \frac{q}{s} \rfloor$ or $\abs{h^{-1}(j)} = \lceil \frac{q}{s} \rceil$.
\end{definition}


For any $\*X \in [q]^V$, define the projection \emph{image} $\*Y \in [s]^V$ of $\*X$ by
\begin{align*}
\forall v \in V,\quad Y_v = h(X_v).
\end{align*}
For simplicity, we often denote $\*Y = h(\*X)$, and for any subset $\Lambda \subseteq V$, we denote $\*Y_{\Lambda} = h(\*X_\Lambda)$.

Given a projection scheme, the following \emph{projected distribution} can be naturally defined.
\begin{definition}[projected distribution]
Given a projection scheme $h$, the projected distribution $\nu$ is the distribution of $\*Y = h(\*X)$, where $\*X \sim \mu$.	
\end{definition}

Given an image of the projection, we can define the following \emph{conditional distribution} over $[q]^V$.
\begin{definition}[conditional distribution]
Let $\Lambda \subseteq V$ be a subset of vertices.
Given a (partial) image $\sigma_{\Lambda} \in [s]^{\Lambda}$, the conditional distribution $\mu^{\sigma_\Lambda}$ is the distribution of $\*X \sim \mu$ conditional on $h(\*X_\Lambda) = \sigma_\Lambda$.
\end{definition}
\noindent
By definition, $\mu^{\sigma_{\Lambda}}$ is a distribution over $[q]^V$.
We use $\mu^{\sigma_\Lambda}_S$ to denote the marginal distribution on $S \subseteq V$ projected from $\mu^{\sigma_\Lambda}$, and we simply denote $\mu^{\sigma_\Lambda}_{\{v\}}$ by $\mu^{\sigma_\Lambda}_{v}$.

\subsection{The sampling algorithm}
In this section and what follows, we always assume that all vertices in $V$ are labeled by $\{0,1,\ldots,n-1\}$.
We also fix the parameter $s = \ctp{\sqrt{q}}$.
Given a projection scheme $h$ with image size $s$,
our sampling algorithm first samples $\*Y \in [s]^V$ from the projected distribution $\nu$, 
and then uses it to sample a random hypergraph colouring from the conditional distribution $\mu^{\*Y}$.
The pseudocode is given in~\Cref{alg:main}.

\begin{algorithm}[H]
\caption{\textsf{Sampling algorithm for hypergraph colouring}}\label{alg:main}
\KwIn{A hypergraph $H=(V,\+E)$, a set of colours $[q]$, an error bound $0<\epsilon<1$, and a balanced projection scheme $h:[q] \to [s]$, where $s = \ctp{\sqrt{q}}$}
\KwOut{A random colouring $\*X \in [q]^V$} 
sample $\*Y \in [s]^V$ uniformly at random\label{line-scan-start}\; 
\For{$t$ from $1$ to $T=\lceil 50n \log \frac{2n\Delta}{\epsilon}\rceil $}{
	let $v$ be the vertex with label $(t\mod n)$\label{line-choose-v}\;
    $X'_v \gets \sample\tp{H,h,\{v\},\*Y_{V \setminus \{v\}}, \frac{\epsilon}{4T}}$\; \label{line-scan-1}
    \Comment{The $\sample$ subroutine is given in \Cref{alg:sample}.}
	$Y_v \gets h(X'_v)$\label{line-scan-2}\;
}
\Return{$\*X \gets  \sample\tp{H,h,V,\*Y, \frac{\epsilon}{4T}}$\;\label{line-last}}
\end{algorithm}

The main ingredient of~\Cref{alg:main} is the part that samples $\*Y$ (\Cref{line-scan-start} to~\Cref{line-scan-2}).
It is basically a \emph{systematic scan} version of the Glauber dynamics for $\nu$.
In order to update the state of a particular vertex, we invoke a subroutine $\sample{}$, given in~\Cref{alg:sample}, to sample $X_v'$ first from the distribution conditional on $\*Y_{V\setminus\{v\}}$.
Also, $\sample{}$ is used to generate the random colouring conditional on $\*Y$ in \Cref{line-last}.
The subroutine $\sample{}$ in fact returns an approximate sample with high probability. 
Here we have to settle with some small error because exactly calculating the conditional distribution is intractable.
To implement $\sample{}$, we use standard rejection sampling, which is described in~\Cref{alg:rej-sample}.
Showing the correctness and efficiency of~\Cref{alg:sample} and~\Cref{alg:rej-sample} is one of our main contributions.

In the following we flesh out the outline above.
Let $\Lambda \subseteq V$ and $\*Y_{\Lambda} \in [s]^\Lambda$. 
Note that during the execution of \Cref{alg:main}, $\*Y_{\Lambda}$ is a random input to $\sample$.
Let $S \subseteq V$ and $\zeta \in (0,1)$. 
The subroutine $\sample\tp{H,h,S,\*Y_{\Lambda},\zeta}$ in \Cref{alg:main} returns a random sample $\*X_S \in [q]^S$ such that with probability at least $1-\zeta$, the total variation distance between $\*X_S$ and $\mu^{\*Y_\Lambda}_S$ is at most $\zeta$, where the probability is taken over the randomness of the input $\*Y_{\Lambda}$. 

In the $t$-th step of the systematic scan in~\Cref{alg:main}, we pick the vertex $v$ with label $(t \mod n)$, 
and use \Cref{line-scan-1} and \Cref{line-scan-2} to update the value of $Y_v$.
Ideally, we want to resample the value of $Y_v$ according to the conditional distribution $\nu^{\*Y_{V \setminus \{v\} }}_v$, where $\nu$ is the distribution projected from $\mu$.
However, exactly computing the conditional distribution is not tractable, and we approximate it by projecting from the random sample $X'_v \in [q]$ given by $\sample$ in \Cref{line-scan-1}.
%
%
It is straightforward to verify that $Y_v$ approximately follows the law of $\nu^{\*Y_{V \setminus \{v\} }}_v$ as long as $X'_v$ approximately follows the law of $\mu^{\*Y_{V \setminus \{v\} }}_v$.
In the last step, we use $\sample$ to draw approximate samples from the conditional distribution $\mu^{\*Y}$.

We explain the details of $\sample\tp{H,h,S,\*Y_{\Lambda},\zeta}$ next.
First we need some notations. Given a partial image $\*Y_\Lambda$, we say an hyperedge $e \in \+E$ is satisfied by $\*Y_\Lambda$ if there exists $u,v\in e \cap \Lambda$ such that $Y_u \neq Y_v$. In other words, for all $\*X \in [q]^V$ such that $\*Y_\Lambda = h(\*X_\Lambda)$, the hyperedge $e$ is not monochromatic with respect to $\*X$, and thus $e$ is always ``satisfied'' given $\*Y_\Lambda$.
Let $H^{\*Y_\Lambda} = (V,\+E^{\*Y_\Lambda})$ be the hypergraph obtained from $H$ by removing all hyperedges satisfied by $\*Y_\Lambda$.
Let $H^{\*Y_\Lambda}_1,H^{\*Y_\Lambda}_2,\ldots,H^{\*Y_\Lambda}_m$ denote the connected components of $H^{\*Y_\Lambda}$, where $H^{\*Y_\Lambda}_i = (V_i, \+E^{\*Y_\Lambda}_i)$. The following fact is straightforward to verify
\begin{align*}
\mu^{\*Y_\Lambda} = \mu_1^{\*Y_{\Lambda \cap V_1}} \times 	 \mu_2^{\*Y_{\Lambda \cap V_2}} \times \ldots \times  \mu_m^{\*Y_{\Lambda \cap V_m}},
\end{align*}
where 
$\mu_i$ is the uniform distribution over proper $q$-colourings of the sub-hypergraph $H_i^{\*Y_\Lambda}$ 
(namely, $\mu_i^{\*Y_{\Lambda \cap V_i}}$ is the uniform distribution over list colourings of $H_i^{\*Y_\Lambda}$ conditional on $\*Y_{\Lambda \cap V_i}$).
Without loss of generality, we assume $S \cap V_j \neq \emptyset$ for $1 \leq j \leq \ell$. 
To draw a random sample from $\mu^{\*Y_\Lambda}_S$, 
it suffices to draw a random sample from the product distribution $\mu_1^{\*Y_{\Lambda \cap V_1}} \times \mu_2^{\*Y_{\Lambda \cap V_2}} \times \ldots \times  \mu_\ell^{\*Y_{\Lambda \cap V_\ell}}$,
which we will do by drawing from each $\mu_i^{\*Y_{\Lambda \cap V_i}}$ individually using standard rejection sampling (given in \Cref{alg:rej-sample}).


One final detail about~\Cref{alg:sample} and~\Cref{alg:rej-sample} is about their efficiency.
Basically we set some thresholds to guard against two unlikely bad events.
We break out from the normal execution immediately and return an arbitrary random sample if one of the following two bad events occur:
\begin{itemize}
  \item for some $1 \leq i \leq \ell$, $|\+E^{\*Y_\Lambda}_i| > 4 \Delta k^3 \log\tp{\frac{n\Delta}{\zeta}}$;
  \item for some $1 \leq i \leq \ell$, the rejection sampling for $\mu_i^{\*Y_{\Lambda \cap V_i}}$ fails after $R$ trials, where 
\begin{align}\label{eq-def-eta}
  R &\defeq \ctp{10\tp{\frac{n\Delta}{\zeta}}^{\frac{1}{1000\eta}}\log \frac{n}{\zeta}} \quad\quad\quad\quad \text{ and } & \eta &\defeq \frac{1}{\Delta} \tp{\frac{q}{100}}^{\frac{k-3}{2}}.
\end{align}
\end{itemize}
In the analysis (see \Cref{lemma-subroutine-detail}), 
we will show that both of the two bad events above occur with low probability, 
and thus with high probability the $\sample$ subroutine returns an approximate sample with desired accuracy.

\begin{algorithm}[h]
\caption{$\sample\tp{H,h,S,\*Y_{\Lambda},\zeta}$}\label{alg:sample}
\KwIn{A hypergraph $H=(V,\+E)$, a projection scheme $h:[q] \to [s]$, a subset $S\subseteq V$, a (partial) image $\*Y_\Lambda \in [s]^{\Lambda}$ where $\Lambda \subseteq V$, and an error bound $\zeta \in (0,1)$%
}
\KwOut{A random (partial) colouring $\*X_S \in [q]^S$} 
remove all hyperedges in $H$ that are satisfied by $\*Y_{\Lambda}$ to obtain $H^{\*Y_\Lambda} = (V,\+E^{\*Y_\Lambda})$\;
let $H_i=(V_i,\+E^{\*Y_\Lambda}_i)$ for $1\leq i \leq \ell$ be the connected components such that $V_i \cap S \neq \emptyset$\label{line-find-connected}\;
\If{$\exists 1 \leq i \leq \ell$ such that $|\+E^{\*Y_\Lambda}_i| > 4 \Delta k^3 \log\tp{\frac{n\Delta}{\zeta}}$\label{line-bad-1-if}}{
\Return{$\*X_S \in [q]^S$ uniformly at random\;\label{line-bad-1}}
}
\For{$i$ from 1 to $\ell$}{
	$\*X_i \gets \Rejsample(H_i,h, \*Y_{\Lambda \cap V_i}, R)$, where $R = \ctp{10\tp{\frac{n\Delta}{\zeta}}^{\frac{1}{1000\eta}}\log \frac{n}{\zeta}}$\label{line-rj-sample}\;
    \Comment{The $\Rejsample$ subroutine is given in \Cref{alg:rej-sample}.}
	\If{$\*X_i = \perp$}{
		\Return{$\*X_S \in [q]^S$ uniformly at random \label{line-bad-2}\;}
	}
}
\Return{$\*X_S$ where $\*X = \biguplus_{i=1}^\ell \*X_i$\;\label{line-good}}
\end{algorithm}

\begin{algorithm}[h]
\caption{$\Rejsample(H,h, \*Y_{\Lambda}, R)$}\label{alg:rej-sample}
\KwIn{A hypergraph $H=(V,\+E)$, a projection scheme $h:[q] \to [s]$, a (partial) image $\*Y_{\Lambda} \in [s]^\Lambda$ where $\Lambda \subseteq V$ and an integer $R$}
\KwOut{A random colouring $\*X \in [q]^V$ or a special symbol $\perp$} 
for each $v \in V$, let $Q_v \gets h^{-1}(Y_v)$ if $v \in \Lambda$, and $Q_v \gets [q]$ if $v \notin \Lambda$\;
\For{$i$ from 1 to $R$}{
	sample $X_v \in Q_v$ uniformly at random for all $v \in V$ and let $\*X = (X_v)_{v \in V}$\;
	\If{$\*X$ is a proper hypergraph colouring of $H$}{
		\Return{$\*X$\;}
	}
}
\Return{$\perp$\;}
\end{algorithm}

\section{Proof of the main theorem}
Let $H=(V,\+E)$ be a simple $k$-uniform hypergraph with maximum degree $\Delta$.
Let $[q]$ be a set of $q$ colours.
Recall $s = \ctp{\sqrt{q}}$, where $s$ is the parameter of projection scheme $h$ (\Cref{definition-projection}).
To construct $h$, we partition $[q]$ into $s$ intervals, 
where the first $(q\bmod s)$ of them contains $\lceil q/s\rceil$ elements each
while the rest contains $\lfloor q/s\rfloor$ elements each. 
For each $i\in[q]$, set
\begin{equation}\label{equ:projection_construct}
h(i)=j\qquad\text{ where $i$ belongs to the $j$-th interval.}
\end{equation}
Note that this $h$ satisfies \Cref{definition-projection}.
In our algorithm, $h$ is implemented as an oracle, 
supporting the following two types of queries. 
\begin{itemize}
  \item Evaluation: given $i$, the oracle returns $h(i)$. 
  \item Inversion: given $j$, the oracle returns a uniform element in $h^{-1}(j)$. 
\end{itemize}
Obviously, each query can be answered in time $O(\log q)$ because of the construction of $h$. 

The next theorem is a stronger form of \Cref{theorem-main}.
It shows that our algorithm can run in time arbitrarily close to linear in $n$, the number of vertices, 
as long as sufficiently many colours are available.
\begin{theorem}\label{theorem-general-main}
The following result holds for any $\delta > 0$ and $0 < \alpha\leq 1$. Given any $\epsilon \in (0,1)$, any $q$-colouring instance on $k$-uniform simple hypergraph $H=(V,E)$ with maximum degree $\Delta$, and a balanced projection scheme, if $k \geq \frac{20(1+\delta)}{\delta}$ and $q \geq 100 \tp{\frac{\Delta}{\alpha}}^{\frac{2+\delta}{k - 4/\delta - 4}}$, \Cref{alg:main} returns a random colouring that is $\epsilon$-close to $\mu$ in total variation distance in time 
$O\tp{\Delta^2 k^5 n\tp{\frac{n\Delta}{\epsilon}}^{\alpha / 100} \log^4 \tp{\frac{n\Delta q}{\epsilon}}}$.
\end{theorem}
\begin{remark}
The parameter $\alpha$ captures the relation between the local lemma condition and the running time of the algorithm. 
If $\alpha$ becomes smaller, the condition is more confined, and the running time is closer to linear. 
In particular, \Cref{theorem-main} is implied by setting $\alpha=1$.
\end{remark}

We need two lemmas to prove \Cref{theorem-general-main}.
The first lemma analyses the mixing time of the idealised systematic scan. 
Let $\nu$ be the projected distribution.
The idealised systematic scan for $\nu$ is defined as follows. 
Initially, let $\*X_0 \in [s]^V$ be an arbitrary initial configuration.
In the $t$-th step, the systematic scan does the following update steps.
\begin{itemize}
\item Pick the vertex $v \in V$ with label $(t \mod n)$ and let $X_t(V \setminus \{v\} ) \gets X_{t-1}(V \setminus \{v\} )$.
\item Sample $X_t(v) \sim \nu^{\*X_{t-1}(V \setminus \{v\} )}_v$.
\end{itemize}

\begin{lemma}\label{lemma-mixing}
  If $q \geq 40\Delta^{\frac{2}{k-4}}$ and $k \geq 20$, the systematic scan chain $\*P_{scan}$ for $\nu$ is irreducible, aperiodic and reversible with respect to $\nu$. Furthermore, the mixing time satisfies
  \begin{align*}
    \forall 0 < \epsilon <1,\quad T_{\mathrm{mix}}(\*P_{scan},\epsilon) \leq \ctp{50n\log \frac{n\Delta}{\epsilon}}.
  \end{align*}
\end{lemma}

\Cref{lemma-mixing} is shown in \Cref{sec:mixing}.

Our next lemma analyzes the $\sample$ subroutine. 
Let $(\*Y_t)_{t=0}^T$ denote the sequence of random configurations in $[s]^V$ generated by \Cref{alg:main}, where $\*Y_0 \in [s]^V$ is the initial configuration and $\*Y_t$ is the configuration after the $t$-th iteration of the for-loop. 
For any $1\leq t \leq T + 1$, consider the $t$-th invocation of $\sample$ and define the following two bad events:
\begin{itemize}
\item $\+B_{\mathrm{com}}(t)$: in the $t$-th invocation, $\*X_S$ is returned by~\Cref{line-bad-1} in \Cref{alg:sample};
\item $\+B_{\mathrm{rej}}(t)$: in the $t$-th invocation, $\*X_S$ is returned by~\Cref{line-bad-2} in \Cref{alg:sample}.
\end{itemize}
Note that the $(T+1)$-th invocation of the subroutine $\sample$ is in \Cref{line-last} in \Cref{alg:main}.
Let $H=(V,\+E)$ denote the input hypergraph of \Cref{alg:main}. 

\begin{lemma}\label{lemma-subroutine-detail}
For any $1 \leq t \leq T+1$, the $t$-th invocation of the subroutine $\sample\tp{H,h,S,\*Y_{\Lambda},\zeta}$, where $h$ is given by (\ref{equ:projection_construct}), satisfies
\begin{enumerate}
\item \label{property-sub-1} the running time of the subroutine is bounded by $O\left(|S|\Delta^2k^5\left(\frac{n\Delta}{\zeta}\right)^{\frac{1}{1000\eta}}\log^3\left(\frac{n\Delta q}{\zeta}\right)\right)$;
\item \label{property-sub-2} conditional on neither $\+B_{\mathrm{com}}(t)$ nor $\+B_{\mathrm{rej}}(t)$ occurs, the subroutine returns a perfect sample from $\mu^{\*Y_\Lambda}_S$;
\item \label{property-sub-3}
  if $q \geq 100\Delta^{\frac{2}{k-3}}$ and $k \geq 20$, then
    $\Pr[\+B_{\mathrm{rej}}(t)] \leq \zeta$;
\item \label{property-sub-4}
  for any $\delta > 0$, if $k \geq \frac{20(\delta+1)}{\delta}$, $q \geq 100\Delta^{\frac{2+\delta}{k-4/\delta-3}}$, and $H$ is simple, then
    $\Pr[\+B_{\mathrm{com}}(t)] \leq \zeta$.	
\end{enumerate}
\end{lemma}

\Cref{lemma-subroutine-detail} is proved in \Cref{section-sample}~and~\ref{section-comp}.

Now we are ready to prove our main result, \Cref{theorem-general-main}.
\begin{proof}[Proof of \Cref{theorem-general-main}]
First note that the condition in \Cref{theorem-general-main} implies all the conditions in \Cref{lemma-mixing} and \Cref{lemma-subroutine-detail}.
Denote the output of \Cref{alg:main} by $\*X_{\-{alg}}$. 
To prove the correctness of our algorithm, the goal is to show
\[
  \DTV{\*X_{\-{alg}}}{\mu}\leq\epsilon.
\]
We first consider an idealized algorithm 
which, instead of simulating the transitions by the $\sample$ subroutine,
is able to run the ideal Glauber dynamics to obtain $\*Y_{\-{ideal}}$ 
before sampling $\*X_\-{ideal}$ from the distribution $\mu^{\*Y_{\-{ideal}}}$. 
By \Cref{lemma-mixing}, running this systematic scan 
for $T=\lceil 50n\log\frac{2n\Delta}{\epsilon}\rceil$ steps 
ensures $\DTV{\*Y_{\-{ideal}}}{\nu}\leq\frac{\varepsilon}{2}$. 
On the other hand, a perfect sample $\*X\sim \mu$ can be drawn by
sampling $\*Y\sim \nu$ first, followed by sampling $\*X\sim \mu^{\* Y}$ based on that. 
The upper bound on total variation distance allows us to couple the perfect $\*Y$ and $\*Y_{\-{ideal}}$
such that $\*Y\neq \*Y_{\-{ideal}}$ with probability no more than $\frac{\epsilon}{2}$.
Conditional on $\*Y=\*Y_{\-{ideal}}$, 
the samples $\*X$ and $\*X_{\-{ideal}}$ on original distribution can be perfectly coupled.
Together with the coupling lemma (\Cref{lem:coupling_lemma}), we have
\[
  \DTV{\*X_{\-{ideal}}}{\mu}\leq\frac{\epsilon}{2}.
\]

Hereinafter, we couple the idealized algorithm with \Cref{alg:main}. 
The nature of systematic scan warrants that both algorithms pick the same vertex in the same step on \Cref{line-choose-v}.
We then try to couple the vertex update as much as possible. 
That is, at Step $t$, if none of $\+B_{\mathrm{com}}(t)$ or $\+B_{\mathrm{rej}}(t)$ happens, 
then the output of $\sample$ subroutine at \Cref{line-scan-1} in \Cref{alg:main} is perfect, 
and hence we can couple it with the idealized systematic scan perfectly. 
The remaining coupling error emerges from the occurrence of $\+B_{\mathrm{com}}(t)$ or $\+B_{\mathrm{rej}}(t)$. 
By the coupling lemma (\Cref{lem:coupling_lemma}) and \Cref{lemma-subroutine-detail}, we have
\[
  \DTV{\*X_{\-{alg}}}{\*X_{\-{ideal}}}\leq\Pr\left[\bigvee_{i=1}^{T}\left(\+B_{\mathrm{com}}(t)\lor \+B_{\mathrm{rej}}(t)\right)\right]=2T\zeta=\frac{\epsilon}{2}
\]
where the last equality is due to the selection of $\zeta$ in \Cref{alg:main}.
Finally, a straightforward application of triangle inequality yields
\[
  \DTV{\*X_{\-{alg}}}{\mu}\leq\DTV{\*X_{\-{alg}}}{\*X_{\-{ideal}}}+\DTV{\*X_{\-{ideal}}}{\mu}=\epsilon
\]
as desired. 

There are $T+1$ invocations to the $\sample$ subroutine in total, 
with the first $T$ calls each costing
\[
  T_{\mathsf{step}}\defeq O\left(\Delta^2k^5\left(\frac{n\Delta}{\epsilon/4T}\right)^{\frac{1}{1000\eta}}\log^3\left(\frac{n\Delta q}{\epsilon/4T}\right)\right)
\]
and the final call on \Cref{line-last} costing
\[
  T_{\mathsf{final}}\defeq O\left(n\Delta^2k^5\left(\frac{n\Delta}{\epsilon/4T}\right)^{\frac{1}{1000\eta}}\log^3\left(\frac{n\Delta q}{\epsilon/4T}\right)\right). 
\]
Summing up, the total running time is
\begin{equation}\label{equ:total_running_time}
  T_{\mathsf{total}}=T\cdot T_{\mathsf{step}}+T_{\mathsf{final}}=O\left((T+n)\Delta^2k^5\left(\frac{n\Delta}{\epsilon/4T}\right)^{\frac{1}{1000\eta}}\log^3\left(\frac{n\Delta q}{\epsilon/4T}\right)\right)
\end{equation}
where
\begin{equation}\label{equ:running_time_expr1}
  T=50n\log\frac{2n\Delta}{\epsilon}\qquad\text{ and }\qquad\eta=\frac{1}{\Delta}\left(\frac{q}{100}\right)^{\frac{k-3}{2}}. 
\end{equation}
Note that the condition $q\geq100\tp{\frac{\Delta}{\alpha}}^{\frac{2+\delta}{k-4/\delta-4}}$ implies
\begin{align*}
  \eta=\frac{1}{\Delta}\left(\frac{q}{100}\right)^{\frac{k-3}{2}}\geq\frac{1}{\Delta}\left(\tp{\frac{\Delta}{\alpha}}^{\frac{2+\delta}{k-4/\delta-4}}\right)^{\frac{k-3}{2}}\geq\frac{1}{\alpha}\Delta^{\frac{(k-3)(1+\delta/2)}{k-4/\delta-4}-1}\geq\frac{1}{\alpha}
\end{align*}
and hence
\begin{equation}\label{equ:running_time_expr2}
  \left(\frac{n\Delta}{\epsilon/4T}\right)^{\frac{1}{1000\eta}}\leq\left(\frac{200n^2\Delta\log\frac{2n\Delta}{\epsilon}}{\epsilon}\right)^{\alpha/1000}=O\left(\left(\frac{n\Delta}{\epsilon}\right)^{\alpha/100}\right). 
\end{equation}
Plugging (\ref{equ:running_time_expr1}) and (\ref{equ:running_time_expr2}) back into (\ref{equ:total_running_time}), we get
\[
  T_{\mathsf{total}}=O\left(\Delta^2k^5n\left(\frac{n\Delta}{\epsilon}\right)^{\alpha/100}\log^4\left(\frac{n\Delta q}{\epsilon}\right)\right)
\]
as desired. 
\end{proof}

\section{Analysis of the \texorpdfstring{$\sample$}{Sample} subroutine}
\label{section-sample}
In this section, we analyse the subroutine $\sample$ and  prove \Cref{lemma-subroutine-detail}.
Properties~\ref{property-sub-1},~\ref{property-sub-2}, and~\ref{property-sub-3} in \Cref{lemma-subroutine-detail} can be proved using techniques developed in~\cite{FGYZ21,FHY20}. 
The proofs are given in \Cref{sec-subroutine-1} and \Cref{sec-subroutine-2}. 
We remark that proofs of the first three properties in \Cref{lemma-subroutine-detail} hold for general hypergraphs, not necessarily simple hypergraphs. 
It is property~\ref{property-sub-4} that requires a simple hypergraph as the input. 
The proof of property~\ref{property-sub-4} is quite involved and is left to~\Cref{section-comp}. 


\subsection{Proof of running time and correctness}\label{sec-subroutine-1}

\begin{proof}[Proof of Property~\ref{property-sub-1} and \ref{property-sub-2}, \Cref{lemma-subroutine-detail}]
Property~\ref{property-sub-2} is straightforwardly implied by the nature of rejection sampling. We now deal with Property~\ref{property-sub-1}. 

Assume all hypergraphs 
are stored as incidence lists. 
We first calculate the time cost of \Cref{line-find-connected}. 
Starting from each $v\in S$, we perform depth-first search (DFS) on $H$, 
and for each edge we encounter, we can check whether it is in $H^{Y_\Lambda}$ in time $O(k)$. 
This procedure can work simultaneously with \Cref{line-bad-1-if}, 
that once the current component reaches size $4\Delta k^3\log\left(\frac{n\Delta}{\zeta}\right)$, 
the subroutine exits in \Cref{line-bad-1}. 
The number of visits by DFS itself will be upper-bounded by 
the number of edges times maximum edge degree which is no larger than $\Delta k$. 
In all, the time complexity of DFS has a crude upper bound
\[
  T_{\mathsf{DFS}}=O\left(|S|\cdot k\cdot 4\Delta k^3\log\left(\frac{n\Delta}{\zeta}\right)\cdot \Delta k\right)=O\left(|S|\Delta^2 k^5\log\left(\frac{n\Delta}{\zeta}\right)\right).
\]
For the time cost of \Cref{line-rj-sample}, be aware $\ell$ is at most $|S|$. 
Suppose the cost of sampling a uniformly random colour from a colour list $Q \subseteq [q]$ is $O(\log q)$.
Each invocation of $\Rejsample$ contains $R$ rounds, 
each of which colours the subgraph $H_i$ and check if it is a proper colouring. 
The cost depends to the number of vertices in $H_i$, 
which is upper-bounded by $k\cdot 4\Delta k^3\log\left(\frac{n\Delta}{\zeta}\right)$. 
The total cost is then
\[
  T_{\mathsf{Rej}}=O\left(|S|\cdot R\cdot \Delta k^4\log\left(\frac{n\Delta}{\zeta}\right)\log q\right)\leq O\left(|S|\Delta k^4\left(\frac{n\Delta}{\zeta}\right)^{\frac{1}{1000\eta}}\log^3\left(\frac{n\Delta q}{\zeta}\right)\right). 
\]
The total running time of $\sample$ is hence given by
\[
  T_{\sample}=T_{\mathsf{DFS}}+T_{\mathsf{Rej}}=O\left(|S|\Delta^2 k^5\left(\frac{n\Delta}{\zeta}\right)^{\frac{1}{1000\eta}}\log^3\left(\frac{n\Delta q}{\zeta}\right)\right). \qedhere
\]
\end{proof}


\subsection{Bound the probability of \texorpdfstring{$\+B_{\mathrm{rej}}(t)$}{Brej(t)}}\label{sec-subroutine-2}

\begin{proof}[Proof of Property~\ref{property-sub-3}, \Cref{lemma-subroutine-detail}]
By the definition of~$\eta$ in~\eqref{eq-def-eta} and the condition in \Cref{lemma-subroutine-detail}, it holds that
\begin{align*}
  q = 100(\eta\Delta)^{\frac{2}{k-3}}, \quad \eta \geq 1, \quad  \text{and} \quad q \ge 100.	
\end{align*}
Consider \Cref{line-rj-sample} in \Cref{alg:sample}.
In the rejection sampling, the input is a hyperedge $H=(V,\+E)$ with at most $4\Delta k ^3 \log \tp{\frac{n\Delta}{\zeta}}$ hyperedges.
The size of the color list for each vertex $v \in V$ satisfies 
\begin{align*}
\abs{Q_v} \geq \ftp{\frac{q}{s}} = \ftp{\frac{q}{\ctp{q}}} \overset{(\ast)}{\geq} \frac{4}{5}\sqrt{q}, 
\end{align*}
where inequality $(\ast)$ holds because $q \geq 100$. 

Let $\+D$ denote the product distribution that each $v \in V$ samples a colour from $Q_v$ uniformly at random.
For each hyperedge $e \in \+E$, let $\+B_e$ denote the bad event that $e$ is monochromatic.
Note that $|Q_v| \leq q$ for all $v \in V$. 
We have for any $e \in \+E$,
\begin{align*}
\Pr_{\+D}[\+B_e] \leq \frac{q}{(\frac{4}{5}\sqrt{q})^{k-1}} = \tp{\frac{5}{4}}^{k-1}q^{\frac{3-k}{2}}	=  \tp{\frac{5}{4}}^{k-1} 100^{\frac{3-k}{2}} \frac{1}{\eta\Delta} \leq \frac{1}{10000\mathrm{e} k^3\eta \Delta},
\end{align*}
where the last inequality holds because $k \geq 20$. For each $e \in \+E$, define $x(e) = \frac{1}{10000 \eta \Delta k^3}$. Note that $\eta \geq 1$. It is straightforward to verify that 
\begin{align*}
\Pr_{\+D}[\+B_e] \leq x(e) \prod_{e':\+B_{e'} \in \Gamma(B_e)}\tp{1-x(e')}.	
\end{align*}

By Lov\'asz local lemma in~\Cref{thm:lll}, it holds that 
\begin{align*}
\Pr_{\+D}\sqtp{\bigwedge_{e \in \+E}\overline{\+B(e)}} \geq \tp{1 - \frac{1}{10000 \eta \Delta k^3}}^{\Delta k ^3 \log \tp{\frac{n\Delta}{\zeta}}}	\geq \exp\tp{-\frac{\log\tp{\frac{n\Delta}{\zeta}} }{5000\eta}} \geq \tp{\frac{\zeta}{n\Delta}}^{\frac{1}{1000\eta}}.
\end{align*}
The rejection sampling repeats for $R = \ctp{10\tp{\frac{n\Delta}{\zeta}}^{\frac{1}{1000\eta}}\log \frac{n}{\zeta}}$ times. 
Hence, the probability that the rejection sampling fails on one connected component is at most 
\begin{align*}
\tp{1 - \tp{\frac{\zeta}{n\Delta}}^{\frac{1}{1000\eta}}}^{R} \leq \exp\tp{-10\log\frac{n}{\zeta}} \leq \tp{\frac{\zeta}{n}}^2. 
\end{align*}
Since there are at most $n$ connected components, by a union bound, we have
\begin{align*}
  \Pr[\+B_{\mathrm{rej}}(t)]&\leq \zeta. \qedhere
\end{align*}
\end{proof}

\section{Analysis of connected components}
\label{section-comp}
In this section, we prove Property~\ref{property-sub-4} in \Cref{lemma-subroutine-detail}.
We assume that the input hypergraph $H$ is simple in this section.
Fix $1\leq t \leq T+1$. Consider the $t$-th invocation of the subroutine $\sample$. 
If $1\leq t \leq T$, we use $v_t$ to denote the vertex picked by the $t$-th step of the systematic scan, i.e. $v_t$ is the vertex with label $(t \mod n)$.
Recall that $\*Y_t \in [s]^V$ is the random configuration generated by \Cref{alg:main} after the $t$-th iteration of the for-loop.
Denote
\begin{align}\label{eq-def-Y}
\Lambda = \begin{cases}
 V \setminus \{v_t\} &\text{if } 1\leq t \leq T\\
 V &\text{if } t = T+1
 \end{cases}\quad and \quad
 \*Y  = \*Y_{t-1}(\Lambda),
\end{align}
so that the input partial configuration to $\sample$ is $\*Y$ (see \Cref{alg:main}).
Hence, we consider the subroutine $\sample\tp{H,h,S,\*Y,\zeta}$, where $\*Y \in [s]^\Lambda$ is a random configuration.

Let $H = (V, \+E)$ denote the input simple hypergraph.
Since $\*Y \in [s]^\Lambda$ is a random configuration, $H^{\*Y}$ is a random hypergraph, where $H^{\*Y}$ is obtained by removing all the hyperedges in $H$ satisfied by $\*Y$. 
Fix an arbitrary vertex $v \in V$.
We use $H_v^{\*Y} = (V^{\*Y}_v, \+E^{\*Y}_v)$ to denote the connected component in $H^{\*Y}$ that contains the vertex $v$. Note that $\+E^{\*Y}_v$ can be an empty set.
A hyperedge $e \in \+E$ is \emph{incident} to $v$ in the hypergraph $H$ if $v \in e$.
We prove the following lemma, which implies property~\ref{property-sub-4}.

\begin{lemma}\label{lemma-bound-compnent}
  For any $\delta > 0$, if $k \geq \frac{20(1+\delta)}{\delta}$, $q \geq 100\Delta^{\frac{2+\delta}{k-4/\delta-3}}$, and $H$ is simple,
  then for any $v \in V$, any $e$ incident to $v$ in $H$, it holds that 
  \begin{align*}
    \Pr_{\*Y}\left[ e \in \+E^{\*Y}_v \land |\+E^{\*Y}_v| \geq 4 \Delta k^3 \log\tp{\frac{n\Delta}{\zeta}}\right] \leq \frac{\zeta}{n\Delta}.	
  \end{align*}
\end{lemma}

We now show that property~\ref{property-sub-4} is a corollary of \Cref{lemma-bound-compnent}. 
Since there are at most $\Delta$ hyperedges incident to $v$, by a union bound, we have for all $v \in V$,
\begin{align*}
\Pr_{\*Y}\sqtp{|\+E^{\*Y}_v| \geq 4 \Delta k^3 \log\tp{\frac{n\Delta}{\zeta}}} \leq  \sum_{e \ni v} \Pr_{\*Y}\left[ e \in \+E^{\*Y}_v \land |\+E^{\*Y}_v| \geq 4 \Delta k^3 \log\tp{\frac{n\Delta}{\zeta}}\right] \leq \frac{\zeta}{n}.
\end{align*}
By a union bound over all vertices $v \in V$, we have
\begin{align*}
\Pr_{\*Y}\sqtp{\exists v \in V \text{ s.t. } |\+E^{\*Y}_v| \geq 4 \Delta k^3 \log\tp{\frac{n\Delta}{\zeta}}} \leq \zeta.	
\end{align*}
This implies the property~\ref{property-sub-4} in \Cref{lemma-subroutine-detail}.
The rest of this section is dedicated to the proof of \Cref{lemma-bound-compnent}.

\subsection{Proof of \texorpdfstring{\Cref{lemma-bound-compnent}}{Lemma \ref{lemma-bound-compnent}}}
Denote by $L_H = (V_L,E_L) = \Lin(H)$ the line graph of $H$ (recall \Cref{def:line_graph}).
%
%
Let $e$ be the hyperedge in \Cref{lemma-bound-compnent} and let $u=u_e$ be the vertex in $L_H$ corresponding to $e$.
Let $L_H^{\*Y} = (V_L^{\*Y}, E_L^{\*Y})$ denote the line graph of $H^{\*Y}$. Note that $L_H^{\*Y}$  is random, and the randomness of $L_H^{\*Y}$ is determined by the randomness of $\*Y$.
Equivalently, the graph $L_H^{\*Y}$ can be generated as follows:
\begin{itemize}
\item remove all vertices $w \in V_L$ such that the corresponding hyperedges in $H$ are satisfied by $\*Y$; let $V_L^{\*Y} \subseteq V_L$ denote the set of remaining vertices;
\item let $L^{\*Y}_H = L_H[V_L^{\*Y}]$ be the subgraph of $L_H$ induced by $V_L^{\*Y}$.
\end{itemize}

Let $\+C \subseteq V_L$ denote the random set of all vertices in the connected component of $L_H^{\*Y}$ that contains the vertex $u$. 
If $u \notin V_L^{\*Y}$, let $\+C = \emptyset$. 
Define an integer parameter $\theta \defeq \ctp{\frac{4}{\delta}}$.
To prove \Cref{lemma-bound-compnent}, it suffices to show that 
\begin{align}
\label{eq-proof-M-target}
\forall M >\theta, \quad
\Pr_{\*Y}\sqtp{\abs{\+C} \geq M}	\leq \tp{\frac{1}{2}}^{\frac{M}{2\theta k^2 \Delta} - 1}.
\end{align}
This is because $k \geq \frac{20(\delta + 1)}{\delta} > \ctp{\frac{4}{\delta}}+1 = \theta+1$,
and setting $M = 4 \Delta k^3 \log\tp{\frac{n\Delta}{\zeta}}$ proves \Cref{lemma-bound-compnent}.

Define the following collection of subsets 
\begin{align*}
\mathrm{Con}_u(M) \defeq \left\{ C \subseteq V_L \mid u \in C \,\land\, \abs{C} = M \,\land\, L_H[C] \text{ is connected}  \right\}.
\end{align*}
It is straightforward to verify that 
\begin{align*}
\Pr_{\*Y}\sqtp{\abs{\+C} \geq M} \leq \Pr_{\*Y} \sqtp{ \exists\, C \in \mathrm{Con}_u(M) \text{ s.t. } C \subseteq V^{\*Y}_L }.	
\end{align*}
In our proof, we partition the set $\mathrm{Con}_u(M)$ into two disjoint subsets
\begin{align*}
\mathrm{Con}_u(M) = \mathrm{Con}_u^{(1)}(M) \uplus 	\mathrm{Con}_u^{(2)}(M), 
\end{align*}
and we bound the probability separately
\begin{align}
\label{eq-prob-cases}
\Pr_{\*Y}\sqtp{\abs{\+C} \geq M} \leq \Pr_{\*Y} \sqtp{ \exists\, C \in \mathrm{Con}_u^{(1)}(M) \text{ s.t. } C \subseteq V^{\*Y}_L } + 	\Pr_{\*Y} \sqtp{ \exists\, C \in \mathrm{Con}_u^{(2)}(M) \text{ s.t. } C \subseteq V^{\*Y}_L }.
\end{align}

We use \Cref{alg:subcomponent} to partition the set $\mathrm{Con}_u(M)$. Taking as an input any $C \in \mathrm{Con}_u(M)$, \Cref{alg:subcomponent} outputs an integer $\ell = \ell(C)$ and disjoint sets $C_1,C_2,\ldots,C_\ell \subseteq C$.  Let
\begin{align}\label{eq-def-Com12}
\forall C \in \mathrm{Con}_u(M), \quad C \in \begin{cases}
 \mathrm{Con}^{(1)}_u(M) &\text{if } \ell(C) \geq \frac{M}{2 \theta k^2 \Delta};\\
 \mathrm{Con}^{(2)}_u(M) &\text{if } \ell(C) < \frac{M}{2 \theta k^2 \Delta}.\\
 \end{cases}
\end{align}
%
We remark that \Cref{alg:subcomponent} is only used for analysis, and we do not need to implement this algorithm.

\begin{algorithm}[H]
\caption{\textsf{2-block-tree generator}}\label{alg:subcomponent}
\KwIn{the parameter $\delta \in (0,1)$ in \Cref{lemma-bound-compnent}, the line graph $L_H$, an integer $M>\theta$, a vertex $u$ in $L_H$, and a subset $C \in \mathrm{Con}_u(M)$}
\KwOut{an integer $\ell$ and connected subgraphs $C_1,\cdots,C_\ell\subseteq C$ 
}  
let $G = L_H[C] = (C,E_C)$ be the subgraph of $L_H$ induced by $C$\;
$\theta \gets \ctp{\frac{4}{\delta}}$, $\ell \gets 0$, $V \gets C$\;
\While{$|V| \geq \theta$}{
	$\ell \gets \ell + 1$\;
	\textbf{if} $\ell = 1$ \textbf{then} $u_\ell \gets u$\;
	\textbf{if} $\ell > 1$ \textbf{then}  let $u_\ell$ be an arbitrary vertex in $\Gamma_G(C \setminus V )$\label{line-pick-u}\;
	let $C_\ell \subseteq V$ be an arbitrary connected subgraph in $G$ such that $|C_\ell| = \theta$ and $u_\ell \in C_\ell$\label{line-pick-c}\;
	$V \gets V \setminus (C_\ell \cup \Gamma_G(C_\ell))$\label{line-delete-neighbour}\;
	\For{each connected component $G'=(V',E')$ in $G[V]$ such that $|V'| < \theta$\label{line-delete-small-if}}{
		$V \gets V \setminus V'$\label{line-delete-small}\;
	}
}
\Return{$\ell,\{C_1,C_2,\ldots,C_\ell\}$\;}
\end{algorithm}


In \Cref{line-pick-u} and \Cref{line-pick-c} of \Cref{alg:subcomponent}, we may use a specific rule to choose the vertex $u_\ell$ and the connected subgraph $C_\ell$
(e.g. pick the element with the smallest index according to an arbitrary but predetermined ordering). 
To explain this algorithm concretely, consider the first round of the \textbf{while}-loop running on the graph in \Cref{fig:2_block_tree_example}, with the parameter $\theta$ set to $3$. 

\begin{figure}[h]
  \centering

  \tikzset{every picture/.style={line width=0.75pt}} 

  \begin{tikzpicture}[x=0.75pt,y=0.75pt,yscale=-0.5,xscale=0.5]
  
  \draw    (340,150) -- (235,95) ;
  \draw    (220,90) -- (155,145) ;
  \draw  [fill={rgb, 255:red, 0; green, 0; blue, 0 }  ,fill opacity=1 ] (250,150) .. controls (250,144.48) and (254.48,140) .. (260,140) .. controls (265.52,140) and (270,144.48) .. (270,150) .. controls (270,155.52) and (265.52,160) .. (260,160) .. controls (254.48,160) and (250,155.52) .. (250,150) -- cycle ;
  \draw  [color={rgb, 255:red, 0; green, 0; blue, 0 }  ,draw opacity=1 ][fill={rgb, 255:red, 0; green, 0; blue, 0 }  ,fill opacity=1 ] (290,105) .. controls (290,99.48) and (294.48,95) .. (300,95) .. controls (305.52,95) and (310,99.48) .. (310,105) .. controls (310,110.52) and (305.52,115) .. (300,115) .. controls (294.48,115) and (290,110.52) .. (290,105) -- cycle ;
  \draw  [fill={rgb, 255:red, 0; green, 0; blue, 0 }  ,fill opacity=1 ] (290,195) .. controls (290,189.48) and (294.48,185) .. (300,185) .. controls (305.52,185) and (310,189.48) .. (310,195) .. controls (310,200.52) and (305.52,205) .. (300,205) .. controls (294.48,205) and (290,200.52) .. (290,195) -- cycle ;
  \draw    (160,150) -- (190,150) ;
  \draw    (205,145) -- (225,100) ;
  \draw    (225,200) -- (205,155) ;
  \draw    (225,80) -- (210,50) ;
  \draw    (235,80) -- (250,50) ;
  \draw    (210,250) -- (225,220) ;
  \draw    (250,250) -- (235,220) ;
  \draw    (255,140) -- (235,100) ;
  \draw    (235,200) -- (255,160) ;
  \draw    (290,105) -- (240,90) ;
  \draw    (290,195) -- (240,210) ;
  \draw    (330,150) -- (270,150) ;
  \draw    (295,110) -- (265,145) ;
  \draw    (295,190) -- (265,155) ;
  \draw    (305,110) -- (335,140) ;
  \draw    (345,145) -- (390,115) ;
  \draw    (345,155) -- (385,190) ;
  \draw    (390,120) -- (390,180) ;
  \draw    (200,285) -- (205,270) ;
  \draw    (220,285) -- (215,270) ;
  \draw    (240,285) -- (245,270) ;
  \draw    (260,285) -- (255,270) ;
  \draw    (215,35) -- (220,20) ;
  \draw    (205,35) -- (200,20) ;
  \draw    (255,35) -- (260,20) ;
  \draw    (245,35) -- (240,20) ;
  \draw    (390,100) -- (395,85) ;
  \draw    (395,110) -- (415,100) ;
  \draw    (400,215) -- (390,200) ;
  \draw    (395,190) -- (415,195) ;
  \draw    (220,210) -- (155,155) ;
  \draw  [dash pattern={on 0.84pt off 2.51pt}] (135,141) .. controls (135,137.69) and (137.69,135) .. (141,135) -- (209,135) .. controls (212.31,135) and (215,137.69) .. (215,141) -- (215,159) .. controls (215,162.31) and (212.31,165) .. (209,165) -- (141,165) .. controls (137.69,165) and (135,162.31) .. (135,159) -- cycle ;
  \draw  [fill={rgb, 255:red, 255; green, 255; blue, 255 }  ,fill opacity=1 ] (200,40) .. controls (200,34.48) and (204.48,30) .. (210,30) .. controls (215.52,30) and (220,34.48) .. (220,40) .. controls (220,45.52) and (215.52,50) .. (210,50) .. controls (204.48,50) and (200,45.52) .. (200,40) -- cycle ;
  \draw  [fill={rgb, 255:red, 255; green, 255; blue, 255 }  ,fill opacity=1 ] (240,40) .. controls (240,34.48) and (244.48,30) .. (250,30) .. controls (255.52,30) and (260,34.48) .. (260,40) .. controls (260,45.52) and (255.52,50) .. (250,50) .. controls (244.48,50) and (240,45.52) .. (240,40) -- cycle ;
  \draw  [fill={rgb, 255:red, 255; green, 255; blue, 255 }  ,fill opacity=1 ] (200,260) .. controls (200,254.48) and (204.48,250) .. (210,250) .. controls (215.52,250) and (220,254.48) .. (220,260) .. controls (220,265.52) and (215.52,270) .. (210,270) .. controls (204.48,270) and (200,265.52) .. (200,260) -- cycle ;
  \draw  [fill={rgb, 255:red, 255; green, 255; blue, 255 }  ,fill opacity=1 ] (240,260) .. controls (240,254.48) and (244.48,250) .. (250,250) .. controls (255.52,250) and (260,254.48) .. (260,260) .. controls (260,265.52) and (255.52,270) .. (250,270) .. controls (244.48,270) and (240,265.52) .. (240,260) -- cycle ;
  \draw  [fill={rgb, 255:red, 226; green, 226; blue, 226 }  ,fill opacity=1 ] (140,150) .. controls (140,144.48) and (144.48,140) .. (150,140) .. controls (155.52,140) and (160,144.48) .. (160,150) .. controls (160,155.52) and (155.52,160) .. (150,160) .. controls (144.48,160) and (140,155.52) .. (140,150) -- cycle ;
  \draw  [fill={rgb, 255:red, 121; green, 121; blue, 121 }  ,fill opacity=1 ] (220,90) .. controls (220,84.48) and (224.48,80) .. (230,80) .. controls (235.52,80) and (240,84.48) .. (240,90) .. controls (240,95.52) and (235.52,100) .. (230,100) .. controls (224.48,100) and (220,95.52) .. (220,90) -- cycle ;
  \draw  [fill={rgb, 255:red, 121; green, 121; blue, 121 }  ,fill opacity=1 ] (220,210) .. controls (220,204.48) and (224.48,200) .. (230,200) .. controls (235.52,200) and (240,204.48) .. (240,210) .. controls (240,215.52) and (235.52,220) .. (230,220) .. controls (224.48,220) and (220,215.52) .. (220,210) -- cycle ;
  \draw  [fill={rgb, 255:red, 121; green, 121; blue, 121 }  ,fill opacity=1 ] (330,150) .. controls (330,144.48) and (334.48,140) .. (340,140) .. controls (345.52,140) and (350,144.48) .. (350,150) .. controls (350,155.52) and (345.52,160) .. (340,160) .. controls (334.48,160) and (330,155.52) .. (330,150) -- cycle ;
  \draw  [fill={rgb, 255:red, 226; green, 226; blue, 226 }  ,fill opacity=1 ] (190,150) .. controls (190,144.48) and (194.48,140) .. (200,140) .. controls (205.52,140) and (210,144.48) .. (210,150) .. controls (210,155.52) and (205.52,160) .. (200,160) .. controls (194.48,160) and (190,155.52) .. (190,150) -- cycle ;
  \draw  [fill={rgb, 255:red, 255; green, 255; blue, 255 }  ,fill opacity=1 ] (380,110) .. controls (380,104.48) and (384.48,100) .. (390,100) .. controls (395.52,100) and (400,104.48) .. (400,110) .. controls (400,115.52) and (395.52,120) .. (390,120) .. controls (384.48,120) and (380,115.52) .. (380,110) -- cycle ;
  \draw  [fill={rgb, 255:red, 255; green, 255; blue, 255 }  ,fill opacity=1 ] (380,190) .. controls (380,184.48) and (384.48,180) .. (390,180) .. controls (395.52,180) and (400,184.48) .. (400,190) .. controls (400,195.52) and (395.52,200) .. (390,200) .. controls (384.48,200) and (380,195.52) .. (380,190) -- cycle ;
  \draw    (230,200) -- (230,100) ;
  \draw  [dash pattern={on 0.84pt off 2.51pt}] (317.5,67.5) -- (317.5,232.5) -- (235,150) -- cycle ;
  
  \draw (322,77.4) node [anchor=north west][inner sep=0.75pt]    {$u$};
  \draw (192,7.4) node [anchor=north west][inner sep=0.75pt]    {$\cdots $};
  \draw (232,7.4) node [anchor=north west][inner sep=0.75pt]    {$\cdots $};
  \draw (192,288.4) node [anchor=north west][inner sep=0.75pt]    {$\cdots $};
  \draw (232,288.4) node [anchor=north west][inner sep=0.75pt]    {$\cdots $};
  \draw (385,62.4) node [anchor=north west][inner sep=0.75pt]    {$\cdots $};
  \draw (410,82.4) node [anchor=north west][inner sep=0.75pt]    {$\cdots $};
  \draw (411,182.4) node [anchor=north west][inner sep=0.75pt]    {$\cdots $};
  \draw (386,217.4) node [anchor=north west][inner sep=0.75pt]    {$\cdots $};
  \draw (323,222.4) node [anchor=north west][inner sep=0.75pt]    {$C_{1}$};

  \end{tikzpicture}

  \caption{The example graph where \Cref{alg:subcomponent} runs on. }
  \label{fig:2_block_tree_example}
\end{figure}
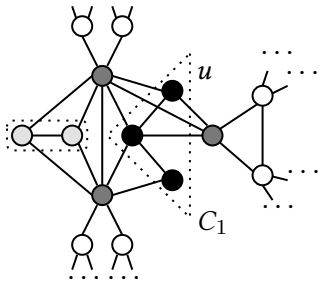

In \Cref{line-pick-c}, the algorithm picks the connected subgraph $C_1$ containing $u$, represented by black circles. 
Then in \Cref{line-delete-neighbour}, the algorithm removes $C_1$, together with its neighbours, depicted by circles in dark grey, from the vertex set $V$. 
Afterwards, the algorithm checks all remaining connected components, and removes those with size less than $\theta=3$ from $V$ in \Cref{line-delete-small}. 
In this example, the algorithm captures and deletes the component in the dotted box. 
Be aware that their neighbours (dark grey circles) have already been removed from $V$. 
As the algorithm goes into the second round of the \textbf{while}-loop, 
the next candidate starting point $u_2$ is selected, as of in \Cref{line-pick-u}, among the vertices depicted by white circles.

To formalize the properties of \Cref{alg:subcomponent}, we begin with the following proposition, which asserts that \Cref{alg:subcomponent} is well defined. The proof is given in \Cref{section-2-block-tree-alg}.

\begin{proposition}\label{proposition-valid}
  Given the input $\delta$, $L_H$, $M$, $u$, and $C \in \mathrm{Con}_u(M)$, \Cref{alg:subcomponent} terminates and generates a unique output.
  Moreover, when \Cref{alg:subcomponent} terminates, $V=\emptyset$.
\end{proposition}

The next proposition, yet of more importance, establishes a few properties of the output of \Cref{alg:subcomponent}. 
They will eventually be used to bound the probabilities on the right hand side (RHS) of~\eqref{eq-prob-cases}.
Before characterising these properties, we introduce a notion called ``2-block-tree''.

\begin{definition} [2-block-tree]\label{def:2-block-tree}
Let $\theta \geq 1$ be an integer.
Let $G=(V,E)$ be a graph.
A set $\{C_1,C_2,\ldots,C_\ell\}$ is a \emph{2-block-tree} with block size $\theta$ and tree size $\ell$ in $G$ if 
\begin{enumerate}[label =\color{blue} (B\arabic*)]
  \item for any $1 \leq i \leq \ell$,  $C_i \subseteq V$, $\abs{C_i} = \theta$, and the induced subgraph $G[C_i]$ is connected; \label{def-2-block-tree-1}
\item for any distinct $1 \leq i,j \leq \ell$, $\dist_G(C_i,C_j) \geq 2$; \label{def-2-block-tree-2}
\item $\{C_1,\cdots,C_\ell\}$ is connected on $G^2$. (Recall \Cref{def:graph_power} of graph powers.) \label{def-2-block-tree-3}
\end{enumerate}
\end{definition}

One can easily observe that the notion of 2-block-trees is a generalisation of 2-trees in~\cite{Alon91} by setting $\theta = 1$. The output of \Cref{alg:subcomponent} is a 2-block-tree in $L_H$. This explains the name  ``2-block-tree generator''.

\begin{proposition}
\label{proposition-component}
The output $\{C_1,C_2,\ldots,C_\ell\}$ of \Cref{alg:subcomponent} satisfies that
\begin{enumerate}
\item $\{C_1,C_2,\ldots,C_\ell\}$ is a 2-block-tree in $L_H$ with block size $\theta$ satisfying $u \in C_1$ and $\cup_{i=1}^\ell C_i \subseteq C$;\label{eq-block-tree-1}
\item if all vertices in $\Gamma_G(C_i)$ are removed from $G$, where $G = L_H[C]$, 
  then the resulting graph $G[C']$ is a collection of connected components whose sizes are at most $\theta$, 
  where $C' = C \setminus (\cup_{i=1}^\ell \Gamma_G(C_i) )$. \label{eq-block-tree-2}
\end{enumerate}
\end{proposition}

In \Cref{proposition-component}, Item~\ref{eq-block-tree-1} is stated with respect to the line graph $L_H$, but Item~\ref{eq-block-tree-2} is stated with respect to the induced subgraph $L_H[C]$. The proof of \Cref{proposition-component} is also given in \Cref{section-2-block-tree-alg}.

Finally, to bound the probabilities on the RHS of~\eqref{eq-prob-cases}, we need the following lemma about the random configuration $\*Y \in [s]^\Lambda$. The proof of \Cref{lemma-random-Y} is given in~\Cref{section-local-uniform}.

\begin{lemma}\label{lemma-random-Y}
If $\lfloor q/s \rfloor^k \geq 2\mathrm{e}qk\Delta$, then for any $R \subseteq \Lambda$, any $\sigma \in [s]^R$, it holds that 
\begin{align*}
\Pr[\*Y_R = \sigma] \leq  \tp{\frac{1}{s}+\frac{1}{q}}^{\abs{R}}\exp\tp{\frac{\abs{R}}{k}}.
\end{align*}
\end{lemma}
The following result is a straightforward corollary of \Cref{lemma-random-Y}.
\begin{corollary}\label{corollary-random-Y}
  Let $\delta > 0$ and $R_1,R_2,\ldots,R_\ell \subseteq \Lambda$ be disjoint subsets.
  For each $1\leq i \leq \ell$, let $\+S_i \subseteq [s]^{R_i}$ be a subset of configurations (namely an event).
If $k \geq \frac{20(\delta+1)}{\delta}$ and $q \geq 100\Delta^{\frac{2+\delta}{k-4/\delta-3}}$, then it holds that 
\begin{align*}
\Pr\sqtp{\bigwedge_{i=1}^\ell \tp{ \*Y_{R_i} \in \+S_i}} \leq  \prod_{i=1}^\ell \abs{\+S_i}\tp{\frac{1}{s}+\frac{1}{q}}^{\abs{R_i}} \exp\tp{\frac{\abs{R_i}}{k}}.
\end{align*}
\end{corollary}

\begin{proof}
Let $R = R_1 \uplus R_2\uplus \ldots \uplus R_\ell$.
Note that $\bigwedge_{i=1}^\ell \tp{ \*Y_{R_i} \in \+S_i}$ if and only if $\*Y_R \in \+S_1 \otimes \+S_2 \otimes \ldots \otimes \+S_\ell$, where
\begin{align*}
 \+S_1 \otimes \+S_2 \otimes \ldots \otimes \+S_\ell \defeq \left\{\sigma \in [s]^R \mid \forall 1 \leq i \leq \ell, \sigma_{R_i} \in \+S_i\right\}.	
\end{align*}
We now verify the condition in \Cref{lemma-random-Y} that $\lfloor q/s \rfloor^k \geq 2\mathrm{e}qk\Delta$. Since $s = \ctp{\sqrt{q}}$ and $q \geq 100$, $\lfloor q/s \rfloor \ge \sqrt{q}/4$.
Thus it suffices to verify $(\sqrt{q}/4)^k \geq 2\mathrm{e}qk\Delta$. The condition in \Cref{corollary-random-Y} implies that $q \geq 100\Delta^{\frac{2}{k-2}}$ and $k \geq 20$, which implies $(\sqrt{q}/4)^k \geq 2\mathrm{e}qk\Delta$. Hence, the condition in \Cref{lemma-random-Y} holds. We have
\begin{align*}
\Pr\sqtp{\bigwedge_{i=1}^\ell \tp{ \*Y_{R_i} \in \+S_i}} = \sum_{\sigma \in  \+S_1 \uplus \+S_2 \uplus \ldots \uplus \+S_\ell}\Pr\sqtp{\*Y_R = \sigma} \leq \prod_{i=1}^\ell \abs{\+S_i}\tp{\frac{1}{s}+\frac{1}{q}}^{\abs{R_i}} \exp\tp{\frac{\abs{R_i}}{k}}. &\qedhere
\end{align*}
\end{proof}

Now, we are ready to bound the probabilities on the RHS of~\eqref{eq-prob-cases}.
We handle the two terms separately:
\begin{align}
  \label{eq-first-case}
  \Pr_{\*Y} \sqtp{ \exists\, C \in \mathrm{Con}_u^{(1)}(M) \text{ s.t. } C \subseteq V^{\*Y}_L } & < \tp{\frac{1}{2}}^{\frac{M}{2\theta k^2 \Delta}};\\
  \label{eq-second-case}
  \Pr_{\*Y} \sqtp{ \exists\, C \in \mathrm{Con}_u^{(2)}(M) \text{ s.t. } C \subseteq V^{\*Y}_L } & < \tp{\frac{1}{2}}^{M}.	
\end{align}
Combining~\eqref{eq-prob-cases} with~\eqref{eq-first-case} and~\eqref{eq-second-case}, we have
\begin{align*}
\Pr_{\*Y}\sqtp{\abs{\+C} \geq M} &\leq \Pr_{\*Y} \sqtp{ \exists\, C \in \mathrm{Con}_u^{(1)}(M) \text{ s.t. } C \subseteq V^{\*Y}_L } + 	\Pr_{\*Y} \sqtp{ \exists\, C \in \mathrm{Con}_u^{(2)}(M) \text{ s.t. } C \subseteq V^{\*Y}_L }\\
&\leq  \tp{\frac{1}{2}}^{\frac{M}{2\theta k^2 \Delta}} + \tp{\frac{1}{2}}^M \leq  \tp{\frac{1}{2}}^{\frac{M}{2\theta k^2 \Delta} - 1}.
\end{align*}
This proves the desired inequality~\eqref{eq-proof-M-target}.

In the next two subsections, we give proofs of \eqref{eq-first-case} and \eqref{eq-second-case}.

\subsubsection{Proof of inequality (\ref{eq-first-case})}
We first prove \eqref{eq-first-case}.
We need to use the following two properties of 2-block-trees, 
the proofs of which are deferred till \Cref{section-2-block-tree}.
\begin{lemma}\label{lemma-block-tree-down}
Let $\theta \geq 1$ be an integer.
Let $G=(V,E)$ be a graph.
For any integer $\ell \geq 2$, any vertex $v \in V$,
if $G$ has a 2-block-tree $\{C_1,C_2,\ldots,C_\ell\}$ with block size $\theta$ and tree size $\ell$ such that $v \in \cup_{i=1}^\ell C_i$, then
there exists an index $1\leq i \leq \ell$ such that $\{C_1,C_2,\ldots,C_\ell\} \setminus \{C_i\}$ is a 2-block-tree in $G$ with block size $\theta$ and tree size $\ell - 1$
and  $v \in \cup_{1\leq j\leq \ell: j\neq i} C_j$.
\end{lemma}

\begin{lemma}\label{lemma-block-tree-count}
Let $\theta \geq 1$ be an integer.
Let $G=(V,E)$ be a graph with maximum degree $d$.
For any integer $\ell \geq 1$, any vertex $v \in V$, the number of 2-block-trees $\{C_1,C_2,\ldots,C_\ell\}$ with block size $\theta$ and tree size $\ell$ such that $v \in \cup_{i=1}^\ell C_i$ is at most $(\theta\mathrm{e}^{\theta}d^{\theta+1})^\ell$.
\end{lemma}

In the rest of this subsection we fix $\ell = \ctp{ \frac{M}{2 \theta k^2 \Delta}}$.
By \eqref{eq-def-Com12}, \Cref{proposition-component}, and \Cref{lemma-block-tree-down}, for any $C \in \mathrm{Con}_u^{(1)}(M)$, there is a 2-block-tree tree $\{C_1,C_2,\ldots,C_\ell\}$ in the line graph $L_H$ with block size $\theta$ and tree size $\ell$ satisfying:
\begin{enumerate}[label =\color{blue} (P\arabic*)]
\item $u \in C_1 \cup C_2 \cup \ldots\cup C_\ell$;\label{2-tree-1}
\item $C_1 \cup C_2 \cup \ldots\cup C_\ell \subseteq C$.\label{2-tree-2}
\end{enumerate}
We denote a 2-block-tree tree with block size $\theta$ and tree size $\ell$ by $(\theta,\ell)$-2-block-tree.
This implies that 
\begin{align}
\label{eq-COM1}
&\Pr_{\*Y} \sqtp{ \exists\, C \in \mathrm{Con}_u^{(1)}(M) \text{ s.t. } C \subseteq V^{\*Y}_L }\notag\\ 
\leq\,& \Pr_{\*Y} \sqtp{ \exists\, (\theta,\ell)\text{-2-block-tree } \{C_1,C_2,\ldots,C_\ell\} \text{ in $L_H$ satisfying~\ref{2-tree-1} s.t. } \forall 1\leq i \leq \ell, C_i \subseteq V^{\*Y}_L}.	
\end{align}
Note that we only need to consider ($\theta,\ell$)-2-block trees satisfying~\ref{2-tree-1}, because \ref{2-tree-2} implies the event that $\forall 1\leq i \leq \ell$, $C_i \subseteq V^{\*Y}_L$.

To bound the probability, we fix a ($\theta,\ell$)-2-block tree $\{C_1,C_2,\ldots,C_\ell\}$ in $L_H$ satisfying~\ref{2-tree-1}.
Fix an index $1 \leq j \leq \ell$.
By \Cref{def:2-block-tree}, $\abs{C_j} = \theta$.
Note that each vertex in $C_j$ represents a hyperedge in the input hypergraph $H = (V,\+E)$.
Let the hyperedges in $C_j$ be $e^j_1,e^j_2,\ldots,e^j_\theta$.
For each $1 \leq t \leq \theta$, we define a subset of vertices $R^j_t \subseteq \Lambda$ (in $H$) by
\begin{align*}
S^j_t  \defeq e^j_t \setminus \tp { \bigcup_{i \in [\theta]: i \neq t} e^j_i} \quad \text{and}\quad R^j_t \defeq S^j_t \cap \Lambda,	
\end{align*}
where $\Lambda$ is defined in~\eqref{eq-def-Y}.
By definition, $R^j_t \subseteq e^j_t$ is a subset of vertices of the input hypergraph $H = (V,\+E)$, and $R^j_t \cap e^j_i = \emptyset$ for any $i \neq t$. 
This implies that $R^j_1,R^j_2,\ldots,R^j_\theta$ are mutually disjoint.
Furthermore, since $H$ is simple and $\abs{\Lambda} \geq \abs{V}-1$, we have
\begin{align}\label{eq-R-lower-bound}
\forall 1 \leq t \leq \theta: \quad \abs{R^j_t} \geq k - (\theta-1) -1 = k - \theta.
\end{align}
The above inequality holds because (1) $|e^j_t| = k$; (2) for each $e^j_i$ with $i \neq t$, the intersection between $e^j_t$ and $e^j_i$ is at most one vertex; and (3) $\abs{\Lambda} \geq \abs{V}-1$.
By \Cref{def:2-block-tree} of 2-block-trees, for $i \neq j$, $\dist_{L_H}(C_i,C_j) \geq 2$.
Let $e \in \+E$ be a hyperedge in $C_i$ and $e' \in \+E$ be a hyperedge in $C_j$, this implies that $e$ and $e'$ are not adjacent in the line graph $L_H$, and thus $e \cap e' = \emptyset$. 
Hence,
\begin{align}\label{eq-R-disjoint}
  (R^j_t)_{1\leq j \leq \ell, 1 \leq t \leq \theta} \text{ are mutually disjoint}.	
\end{align}

We now bound the probability of $C_j \subseteq V^{\*Y}_L$ for all $1\leq j \leq \ell$.
For all $1\leq j \leq \ell$ and $1 \leq t \leq \theta$, since $C_j \subseteq V^{\*Y}_L$, the hyperedge $e^j_t$ is not satisfied by $\*Y$, thus $e^j_t$ is monochromatic with respect to $\*Y$, i.e. for all $v,v' \in e^j_t$, it holds that $Y_v= Y_{v'}$.
Note that $R^j_t \subseteq e^j_t$. We have the following bound
\begin{align}\label{eqn:C-to-R}
\Pr_{\*Y}\sqtp{\forall 1\leq j \leq \ell, C_j \subseteq V^{\*Y}_L} \leq \Pr_{\*Y}\sqtp{\forall 1\leq j \leq \ell, 1 \leq t \leq \theta, R^j_t \text{ is monochromatic w.r.t. } \*Y}.	
\end{align}
Let $\+S^j_t$ be the set of all $s$ monochromatic configurations of $R^j_t$ (i.e. all vertices in $R^j_t$ take the same value $c$, where $c \in [s]$), or more formally,
\begin{align*}
\+S^j_t =\{\sigma \in \{c\}^{R^j_t} \mid c \in [s]\}.	
\end{align*}
In particular, $\abs{\+S^j_t}=s$.
By \Cref{corollary-random-Y},~\eqref{eq-R-lower-bound},~\eqref{eq-R-disjoint}, and~\eqref{eqn:C-to-R}, it holds that
\begin{align*}
\Pr_{\*Y}\sqtp{\forall 1\leq i \leq \ell, C_i \subseteq V^{\*Y}_L} 	 &\leq \Pr_{\*Y}\sqtp{\bigwedge_{j=1}^\ell \bigwedge_{t =1}^\theta \tp{Y_{R^j_t} \in \+S^j_t} }\leq \prod_{i=1}^\ell\prod_{t = 1}^\theta s \tp{\frac{1}{s}+\frac{1}{q}}^{\abs{R^j_t}}\exp\tp{\frac{|R^j_t|}{k}}\\
&\leq s^{\ell \theta} \prod_{i=1}^\ell\prod_{t = 1}^\theta \tp{\frac{1}{s}+\frac{1}{q}}^{\abs{R^j_t}}\exp\tp{\frac{|R^j_t|}{k}} \\
\tp{\text{as $k - \theta \leq |R^j_t|\leq k$}}\quad &\leq (\mathrm{e}s)^{\ell \theta} \tp{\frac{1}{s}+\frac{1}{q}}^{\ell \theta (k-\theta)} = \tp{(\mathrm{e}s)^{\theta} \tp{\frac{1}{s}+\frac{1}{q}}^{\theta (k-\theta)}}^\ell.
\end{align*}

Note that the maximum degree of $L_H$ is no more than $k\Delta$.
By \Cref{lemma-block-tree-count} and a union bound over all possible 2-block-trees, we have
\begin{align}
\label{eq-component-union-bound}
&\Pr_{\*Y} \sqtp{ \exists\, (\theta,\ell)\text{-2-block-tree } \{C_1,C_2,\ldots,C_\ell\} \text{in $L_H$ satisfying~\ref{2-tree-1} s.t. } \forall 1\leq i \leq \ell, C_i \subseteq V^{\*Y}_L}\notag\\
\leq &\, \tp{\theta \mathrm{e}^{2\theta }(k\Delta) ^{\theta + 1}s^\theta \tp{\frac{1}{s} + \frac{1}{q}}^{\theta (k-\theta)}}^\ell \leq \tp{\theta\mathrm{e}^{2\theta }2^{\theta(k-\theta)} (k\Delta) ^{\theta + 1}s^{\theta - \theta(k-\theta)} }^\ell,
\end{align}
where the last inequality uses the fact that $\frac{1}{s} + \frac{1}{q} \leq \frac{2}{s}$.
We will show that 
\begin{align}\label{eqn:base-bound-1/2}
  \theta\mathrm{e}^{2\theta}2^{\theta(k-\theta)} (k\Delta) ^{\theta + 1}s^{\theta - \theta(k-\theta)} \leq \frac{1}{2}.
\end{align}
Recall that $k > \theta+1$, and consequently $\theta(k-\theta) - \theta > 0$.
It implies that
\begin{align*}
  \theta\mathrm{e}^{2\theta}2^{\theta(k-\theta)} (k\Delta) ^{\theta + 1}s^{\theta - \theta(k-\theta)} \leq \frac{1}{2} 
  \quad\Longleftrightarrow\quad 
  s \geq \theta^{\frac{1}{\theta(k-\theta)-\theta}}\mathrm{e}^{\frac{2\theta}{\theta(k-\theta)-\theta}}2^{\frac{\theta(k-\theta)+1}{\theta(k-\theta)-\theta}}(k\Delta)^{\frac{\theta+1}{\theta(k-\theta)-\theta}}.
\end{align*}
Recall that $s = \ctp{\sqrt{q}} \geq q^{1/2}$. 
It suffices to show that
\begin{align*}
  q \geq \theta^{\frac{2}{\theta(k-\theta)-\theta}} \mathrm{e}^{\frac{4\theta}{\theta(k-\theta)-\theta}}2^{\frac{2\theta(k-\theta)+2}{\theta(k-\theta)-\theta}}(k\Delta)^{\frac{2\theta+2}{\theta(k-\theta)-\theta}} 
  = \theta^{\frac{2}{\theta(k-\theta)-\theta}}\mathrm{e}^{\frac{4}{k-\theta-1}}2^{\frac{2(k-\theta)+2/\theta}{k-\theta-1}}(k\Delta)^{\frac{2+2/\theta}{k-\theta-1}}.
\end{align*}
Recall that $\theta = \ctp{\frac{4}{\delta}}$. If $\delta \geq 4$, then $\theta = 1$. In this case, we only need to show that
\begin{align*}
  q \geq \mathrm{e}^{\frac{4}{k-2}}2^{\frac{2k}{k-2}}k^{\frac{4}{k-2}}\Delta^{\frac{2+\delta/2}{k-2}}.
\end{align*}
Otherwise $0 < \delta < 4$, in which case we only need to show that
\begin{align*}
  q > 2\mathrm{e}^{\frac{4}{k-4/\delta-2}}2^{\frac{2k-8/\delta+\delta/2}{k-4/\delta-2}}(k\Delta)^{\frac{2+\delta/2}{k-4/\delta-2}},
\end{align*}
as $\theta^{\frac{2}{\theta(k-\theta)-\theta}}<2$ and $4/\delta\le \theta<4/\delta+1$.
The conditions $k \geq \frac{20(\delta+1)}{\delta}$ and $q \geq 100\Delta^{\frac{2+\delta}{k-4/\delta-3}}$ imply both conditions above. 
This finishes the proof of \eqref{eqn:base-bound-1/2}.
Finally, \eqref{eq-first-case} follows from combining~\eqref{eq-COM1},~\eqref{eq-component-union-bound}, and~\eqref{eqn:base-bound-1/2}.

\subsubsection{Proof of inequality (\ref{eq-second-case})}
We continue to show \eqref{eq-second-case}.
Fix a connected component $C \in \mathrm{Con}_u^{(2)}(M)$. 
We analyse the probability of $C \subseteq V^{\*Y}_L$.
We run \Cref{alg:subcomponent} with the input $C$.
The algorithm outputs an integer $\ell < {\frac{M}{2\theta k^2 \Delta}}$ and a set of connected components $C_1,C_2,\ldots,C_\ell$.
Let $G = L_H[C]$ be the subgraph of $L_H$ induced by $C$.
By \Cref{proposition-component}, after removing all vertices of $\Gamma_G(C_i)$ for all $1 \leq i \leq \ell$, the graph $G$ is decomposed into connected components with vertex sets $D_1,D_2,\ldots,D_m \subseteq C$ such that $\abs{D_i} \leq \theta$ for all $1 \leq j \leq m$.
Note that given $C \in \mathrm{Con}_u^{(2)}(M)$, all the sets $D_1,D_2,\ldots,D_m \subseteq C$ are uniquely determined by \Cref{alg:subcomponent}. We have
\begin{align*}
\Pr_{\*Y}[C \subseteq V^{\*Y}_L] \leq \Pr_{\*Y}\sqtp{ \bigwedge_{j=1}^m \tp{ D_j\subseteq V^{\*Y}_L } }.	
\end{align*}

We then use an analysis similar to the last subsection but focused on the $D_j$'s. For each $1\leq j \leq m$, each vertex in $D_j$ represents a hyperedge in the input hypergraph $H=(V,\+E)$. 
Let $d(j) = \abs{D_j}$.
Let $e^j_1,e^j_2,\ldots,e^j_{d(j)}$ denote the hyperedges in $D_j$.
For each $1 \leq t \leq d(j)$, we define
\begin{align*}
S^j_t  \defeq e^j_t \setminus \tp { \bigcup_{i \in [d(j)]: i \neq t} e^j_i} \quad\text{and}\quad R^j_t \defeq S^j_t \cap \Lambda. 	
\end{align*}
Since $H$ is simple, $\abs{D_j} \leq \theta$, and $\abs{\Lambda} \geq \abs{V}-1$, it holds that
\begin{align}\label{eq-S-lower-bound}
\forall 1 \leq t \leq d(j): \quad \abs{R^j_t} \geq k - (\theta-1) - 1 = k - \theta.	
\end{align}
Next, note that $D_1,D_2,\ldots,D_m \subseteq C$ is a set of disjoint connected components in the induced subgraph $G[D]$, where $D = C \setminus (\cup_{i=1}^\ell \Gamma_G(C_i)) = \cup_{i=1}^mD_i$. For any two distinct $1 \leq i,j\leq m$, $\dist_G(D_i,D_j) \geq 2$, as otherwise $D_i$ and $D_j$ must have been merged into one component.
As $G= L_H[C]$ is a subgraph of $L_H$ induced by $C$, for any two distinct $1 \leq i,j\leq m$, $\dist_{L_H}(D_i,D_j) \geq 2$.
Hence, for any hyperedge $e \in \+E$ in $D_i$, any hyperedge $e' \in \+E$ in $D_j$, it holds that $e \cap e' = \emptyset$.
It implies that
\begin{align}\label{eq-S-disjoint}
(R^j_t)_{1\leq j \leq m, 1 \leq t \leq d(j)} \text{ are mutually disjoint}.	
\end{align}
Again, let $\+S^j_t$ denote the set of all $s$ monochromatic configurations of $R^j_t$ (i.e. all vertices in $R^j_t$ taking the same value $c$, where $c \in [s]$).
By \Cref{corollary-random-Y} and~\eqref{eq-S-disjoint}, it holds that
\begin{align*}
\Pr_{\*Y}[C \subseteq V^{\*Y}_L] &\leq \Pr_{\*Y}\sqtp{ \bigwedge_{j=1}^m \tp{ D_j\subseteq V^{\*Y}_L } } \leq \Pr_{\*Y}\sqtp{\bigwedge_{j=1}^m \bigwedge_{t=1}^{d(j)} \tp{R^j_t \subseteq V^{\*Y}_L} } = \Pr_{\*Y}\sqtp{\bigwedge_{j=1}^m \bigwedge_{t=1}^{d(j)} \tp{Y_{R^j_t} \in \+S^j_t}}\\
&\leq \prod_{j=1}^m \prod_{t= 1}^{d(j)}\tp{s \tp{\frac{1}{s}+\frac{1}{q}}^{|R^j_t|} \exp\tp{\frac{|R^j_t|}{k}} } \leq \prod_{j=1}^m \prod_{t= 1}^{d(j)}\tp{\mathrm{e}s \tp{\frac{1}{s}+\frac{1}{q}}^{|R^j_t|}},
\end{align*}
where the last equation holds because $|R^j_t| \leq k$. Define
\begin{align*}
R \defeq \bigcup_{j=1}^m \bigcup_{t=1}^{d(j)}R^j_t	
\end{align*}
as the (disjoint) union of all $R^j_t$. By the lower bound in~\eqref{eq-S-lower-bound}, we have
\begin{align*}
|R| \geq \sum_{j=1}^m \sum_{t = 1}^{d(j)}(k-\theta) = (k-\theta)\sum_{j=1}^m d(j) = (k-\theta)\tp{M-\abs{\bigcup_{i=1}^\ell \Gamma_G(C_i)}},	
\end{align*}
where the last equation holds because $\{D_i\}_{1\le i\le m}$ is a partition of $C \setminus (\cup_{i=1}^\ell \Gamma_G(C_i))$ and $|C| = M$.
Note that for any $1\leq i \leq \ell$, $|C_i| = \theta$ and the maximum degree of the line graph $L_H$ is at most $k\Delta$. We have
\begin{align*}
|R| \geq (k-\theta)\tp{M -  \ell \theta k \Delta}	.
\end{align*}
This implies 
\begin{align*}
\Pr_{\*Y}[C \subseteq V^{\*Y}_L] \leq  \prod_{j=1}^m \prod_{t= 1}^{d(j)}\tp{\mathrm{e}s \tp{\frac{1}{s}+\frac{1}{q}}^{|R^j_t|}}	= (\mathrm{e}s)^{\sum_{i=1}^m d(j) } \tp{\frac{1}{s}+\frac{1}{q}}^{|R|} \leq (\mathrm{e}s)^M\tp{\frac{1}{s} + \frac{1}{q}}^{(k-\theta)(M-\ell \theta k \Delta)},
\end{align*}
where we use the fact $\sum_{i=1}^m d(j) \leq M$ in the last inequality. 
Since $C \in \mathrm{Con}_u^{(2)}(M)$, it holds that $\ell < \frac{M}{2 \theta k^2 \Delta}$.
Combining with the fact that $\frac{1}{s} + \frac{1}{q} \leq \frac{2}{s}$, we have
\begin{align*}
\Pr_{\*Y}[C \subseteq V^{\*Y}_L] \leq (\mathrm{e}s)^M\tp{\frac{2}{s}}^{(k-\theta)\tp{M - \frac{M}{2k}}} \leq (\mathrm{e}s)^M\tp{\frac{2}{s}}^{(k-\theta)M}\tp{\frac{s}{2}}^{\frac{M}{2}}.
\end{align*}

In order to give a rough bound on the number of connected subgraphs containing $u$,
we will use the following well-known result by Borgs, Chayes, Kahn, and Lov\'{a}sz \cite{borgs2013left}.

\begin{lemma}[{\cite[Lemma 2.1]{borgs2013left}}]\label{lem:connect_count}
  Let $G=(V,E)$ be a graph with maximum degree $d$ and $v\in V$ be a vertex. 
  Then the number of connected induced subgraphs of size $\ell$ containing $v$ is at most
  $(ed)^{\ell-1}/2$. 
\end{lemma}

The maximum degree of $L_H$ is at most $k\Delta$.
By \Cref{lem:connect_count}, the number of connected subgraphs of size $M$ containing $u$ in $L_H$ is at most $(\mathrm{e}\Delta k)^{M-1}/2$. 
Hence $\abs{\mathrm{Con}_u^{(2)}(M)} < (\mathrm{e}\Delta k)^M$. By a union bound over all $C \in \mathrm{Con}_u^{(2)}(M)$, we have
\begin{align*}
\Pr_{\*Y} \sqtp{ \exists\, C \in \mathrm{Con}_u^{(2)}(M) \text{ s.t. } C \subseteq V^{\*Y}_L } \leq (\mathrm{e}\Delta k)^M	(\mathrm{e}s)^M\tp{\frac{2}{s}}^{(k-\theta)M}\tp{\frac{s}{2}}^{\frac{M}{2}}= \tp{\mathrm{e}^2s \Delta k \tp{\frac{2}{s}}^{(k-\theta)}}^M \tp{\frac{s}{2}}^{\frac{M}{2}}.
\end{align*}
We claim that
\begin{align*}
\mathrm{e}^2s \Delta k \tp{\frac{2}{s}}^{(k-\theta)} \leq \frac{1}{s}.
\end{align*}
Since $s = \ctp{\sqrt{q}}$, it suffices to show that 
\begin{align*}
  q \geq 	\mathrm{e}^{\frac{4}{k-\theta - 2}} 2^{\frac{2(k-\theta)}{k-\theta-2}}k^{\frac{2}{k-\theta-2}}\Delta^{\frac{2}{k-\theta-2}},
\end{align*}
which is, in turn, implied by $\theta = \ctp{\frac{4}{\delta}}$, $k \geq \frac{20(\delta+1)}{\delta}$ and $q \geq 100\Delta^{\frac{2+\delta}{k-4/\delta-3}}$. Hence, we have
\begin{align*}
\Pr_{\*Y} \sqtp{ \exists\, C \in \mathrm{Con}_u^{(2)}(M) \text{ s.t. } C \subseteq V^{\*Y}_L } \leq \tp{\frac{1}{s}}^M \tp{\frac{s}{2}}^{\frac{M}{2}}	\leq \tp{\frac{1}{2}}^M,
\end{align*}
where the last inequality holds because $s \geq \sqrt{q} \geq 10$.
\subsection{Properties of the 2-block-tree generator}
\label{section-2-block-tree-alg}

We begin with validating \Cref{alg:subcomponent}, namely proving \Cref{proposition-valid}. 

\begin{proof}[Proof of \Cref{proposition-valid}]
We claim that the algorithm always succeeds in \Cref{line-pick-u} and \Cref{line-pick-c}, 
which implies that the size of $V$ strictly decreases in every step and the algorithm halts eventually. 
Moreover, if $\abs{V}<\theta$, then all vertices in $V$ will be removed in \Cref{line-delete-small-if} and \Cref{line-delete-small}.
Also, so long as $u_\ell$ and $C_\ell$ are selected according to some (arbitrary but) deterministic rule, 
the output is deterministic. 

For the claim, first notice that $V\subseteq C$ throughout the algorithm.
For \Cref{line-pick-u}, since $G=L_H[C]$ is connected and $V\neq\emptyset$,
$\Gamma_G(C\setminus V)\neq\emptyset$ and thus $u_\ell$ exists.
For \Cref{line-pick-c}, $C_\ell$ exists as long as the connected component containing $u_\ell$ in $G[V]$ has size at least $\theta$.
In the first iteration of the while-loop, this holds true as $\abs{V}=\abs{C}=M > \theta$ and $G[V]=G$ is connected.
In all iterations thereafter, the size of the component cannot be smaller than $\theta$,
as otherwise it would have been removed in the previous iteration at \Cref{line-delete-small-if} and \Cref{line-delete-small}.
%
\end{proof}

We then prove \Cref{proposition-component}. The following observation will be useful. 

\begin{proposition} \label{prop-algo-dist-2}
Let $\ell>1$ and $u_\ell$ be the vertex selected in \Cref{line-pick-u}. 
Then there exists some $1\leq j<\ell$ such that $\dist_G(C_j,u_\ell)=2$.  
\end{proposition}

\begin{proof}
Assume for contradiction that $\dist_G(C_j,u_\ell)>2$ for all $1\leq j<\ell$. 
Consider the set $V$ when $u_\ell$ is selected.
Because of \Cref{line-pick-u}, 
we can find one of $u_\ell's$ neighbours that is in $C\backslash V$, say $v$. 
Consider the reason why $v$ was removed from $V$. 
If this happened on \Cref{line-delete-neighbour}, 
then there must have been some $i$ such that $v\in C_i$ or $v\in\Gamma_G(C_i)$. 
The former case implies that $u_\ell$ must have been removed from $V$, which is impossible. 
The latter case indicates $\dist_G(C_i,u_\ell)=2$, a contradiction. 
Therefore, $v$ was removed in \Cref{line-delete-small}.
However, this implies that $u_\ell$ would have been removed from $V$ too, because $u_\ell$ and $v$ must have been in the same component~$V'$,
which is also a contradiction.
\end{proof}

\begin{proof}[Proof of \Cref{proposition-component}]
The first part of this proposition requires us to verify that $\{C_1,\cdots,C_\ell\}$ is a $2$-block-tree in $L_H$. 
To do so, we verify Items \ref{def-2-block-tree-1}, \ref{def-2-block-tree-2}, and \ref{def-2-block-tree-3} of \Cref{def:2-block-tree} next.
Notice that what we need to prove here is with respect to $L_H$, instead of $G=L_H[C]$.
\begin{itemize}
  \item Item \ref{def-2-block-tree-1} holds due to how $C_i$ is constructed in \Cref{line-pick-c}.
  \item For Item \ref{def-2-block-tree-2}, 
  we first show $\dist_G(C_i,C_j)\geq 2$. 
  For any $C_i$ generated by \Cref{alg:subcomponent}, 
  it is ensured that $\Gamma_G(C_i)$ gets removed from $V$, 
  and therefore, no vertex in $\Gamma_G(C_i)$ will be in $C_j$ for any other $j$. 
  To show $\dist_{L_H}(C_i,C_j)\geq 2$, note that $G$ is an induced subgraph of $L_H$. 
  Any two vertices of distance more than $1$ in $G$ cannot be neighbours in $L_H$, 
  and this implies $\dist_{L_H}(C_i,C_j)\geq 2$.
  \item To verify \ref{def-2-block-tree-3}, it suffices to show that $\{C_1,\cdots,C_\ell\}$ is connected in $G^2$, because $G$ is a subgraph of $L_H$. 
  This follows from a simple induction.  
  Suppose $\{C_1,\cdots,C_i\}$, in the order of being generated by the algorithm, 
  is connected in $G^2$. 
  The base case of $i=1$ holds since $C_1$ is connected.
  Now consider $C_{i+1}$. 
  By \Cref{prop-algo-dist-2}, there exists some $j$ such that $\dist_G(C_{i+1},C_j)=2$, 
  which implies that $\{C_1,\cdots,C_{i+1}\}$ is connected in $G^2$ as well. 
\end{itemize}

For the second part, suppose towards contradiction that 
there is some connected component $C^*$ in $G[C']$ of size greater than $\theta$. 
All vertices in $C$ must have been removed from $V$ when the algorithm halts, according to \Cref{proposition-valid}.
However, $C^*$ cannot be $C_i$ for any $i$, because $|C_i|=\theta$. 
It cannot contain any vertex in $\Gamma_G(C_i)$ either by the definition of $C'$. 
Thus, no vertex in $C^*$ can be removed in \Cref{line-delete-neighbour}, 
and all vertices in $C^*$ must have been removed from $V$ in \Cref{line-delete-small}. 
Because $C^*$ does not contain any vertex from either $C_i$ or $\Gamma_G(C_i)$, 
it does not split into smaller components whilst the algorithm is executed. 
Thus, the whole $C^*$ must have been removed from $V$ in a single step, 
which means $|C^*|<\theta$, a contradiction.
\end{proof}

\subsection{Property of random configurations}
\label{section-local-uniform}
%
\begin{proof}[Proof of \Cref{lemma-random-Y}]
  Recall that $\*Y \in [s]^\Lambda$, defined in~\eqref{eq-def-Y}, is the configuration at time $t-1$ on $\Lambda$.
  For each vertex $w \in V$, let $t(w)$ denote  $\max_{1\leq t' < t}$ such that vertex $w$ is updated by the systematic scan in the $t'$-th step 
  (i.e.~the label of $w$ is $t' \mod n$), and let $t(w) = 0$ when such $t'$ does not exist.
With this notation $Y_w = Y_{t(w)}(w)$ for all $w \in \Lambda$.
We assume $R = \{w_1,w_2,\ldots,w_{\abs{R}}\}$ such that $t(w_1) \leq t(w_2) \leq \ldots \leq t(w_{\abs{R}})$. 
By the chain rule, we have $\Pr[\*Y_R = \sigma] = \prod_{i = 1}^{\abs{R}} p_i$,
where $p_i=\Pr\sqtp{Y_{w_i} = \sigma_{w_i} \mid  \bigwedge_{j=1}^{i-1} Y_{w_j} = \sigma_{w_j}}$.
We now bound the value of each $p_i$ as follows. 
If $t(w_i) = 0$, then it holds that $p_i \leq \frac{\ctp{q/s}}{q}$. If $t(w_i) > 0$, then in the $t(w_i)$-th iteration, the algorithm first samples $X'_{w_i}$ using $\sample$, and then sets $Y_{w_i} = h(X'_{w_i})$. 
Denote $\*Y' = \*Y_{t(w_i)-1}(V \setminus \{w_i\})$.
There are two sub-cases:
\begin{itemize}
\item if $X'_{w_i}$ is returned by~\Cref{line-bad-1} or \Cref{line-bad-2} in $\sample$, then $X'_{w_i}$ is sampled uniformly at random from $[q]$, which implies that $p_i \leq \frac{\ctp{q/s}}{q}$;
\item if $X'_{w_i}$ is returned by~\Cref{line-good} in $\sample$, by property~\ref{property-sub-2} of \Cref{lemma-subroutine-detail}, 
  $X'_{w_i}$ is sampled from the correct conditional distribution $\mu_{w_i}^{\*Y'}$. 
  Note that for any $\tau \in [s]^{V \setminus \{w_i\}}$, $\mu_{w_i}^{\tau}$ is the marginal distribution induced by a list hypergraph colouring instance
  where the colour list of any $w\neq w_i$ is $h^{-1}(\tau(w))$, where $h$ is the projection scheme, and $w_i$'s colour list is $[q]$.
  By \Cref{definition-projection} of projection schemes,
  for any $w \neq w_i$, $\abs{h^{-1}(\tau(w))}\ge\lfloor q/s \rfloor$.
  In other words, the upper bound on the size of the lists is $q$ and the lower bound is $\lfloor q/s \rfloor$.
  Since $\lfloor q/s \rfloor^k \geq 2\mathrm{e}qk\Delta$, by \Cref{lemma-local-uniform-colour}, it holds that for all $\tau \in [s]^{V \setminus \{w_i\}}, c\in [q]$,
  \begin{align*}
    \Pr\sqtp{X'_w = c \mid \*Y' = \tau \land \text{$X'_{w_i}$ is returned by~\Cref{line-good}}}\leq \frac{1}{q}\exp\tp{\frac{1}{k}},
  \end{align*}
  which implies $p_i \leq \frac{\ctp{q/s}}{q}\exp\tp{\frac{1}{k}}$.
\end{itemize}
Combining all the cases together, we have
\begin{align*}
\Pr[\*Y_R = \sigma] \leq \tp{\frac{\ctp{q/s}}{q}}^{\abs{R}}\exp\tp{\frac{\abs{R}}{k}} \leq \tp{\frac{q/s + 1}{q}}^{\abs{R}}\exp\tp{\frac{\abs{R}}{k}} = \tp{\frac{1}{s}+\frac{1}{q}}^{\abs{R}}\exp\tp{\frac{\abs{R}}{k}}.&\qedhere
\end{align*}

\end{proof}

\subsection{Properties of 2-block-trees}
\label{section-2-block-tree}

In this subsection, we show \Cref{lemma-block-tree-down} and \Cref{lemma-block-tree-count}. 
We begin with the first one, which is a simple observation. 

\begin{proof}[Proof of \Cref{lemma-block-tree-down}]
Given a $2$-block-tree $\{C_1,\cdots,C_\ell\}$ of $G$ and the vertex $v$,
construct the following graph $G_C$. 
Each vertex $u_j$ of $G_C$ corresponds to a block $C_j$, 
and two vertices $u_j,u_{j'}$ are adjacent if and only if $\dist_G(C_j,C_{j'})=2$. 
By the definition of $2$-block-tree, the graph $G_C$ is connected. 
Therefore, we can take an arbitrary spanning tree of it. 
To select the $C_i$ to drop, note that any tree containing at least $2$ vertices
has at least $2$ vertices of degree $1$. 
Therefore, we just choose $u_i$ to be one such vertex where $v\notin C_i$. 
The rest of the tree is still connected, and so is $G_C-u_i$, 
which indicates that $\{C_1,\cdots,C_\ell\}-C_i$ still forms a $2$-block-tree that contains $v$. 
\end{proof}

We proceed to show \Cref{lemma-block-tree-count}. 
We may apply \Cref{lem:connect_count} on $G^2$ due to property \ref{def-2-block-tree-3}.
Unfortunately, this yields roughly $(ed^2)^{\theta\ell}$ and does not suffice for our purpose. 
Here, we give a refined estimation inspired by the original embedding argument of \cite{Sta99,borgs2013left}. 

Let $d':=(ed)^{\theta-1}/2$, which, by \Cref{lem:connect_count}, upper bounds the number of size-$\theta$ connected induced subgraphs containing a given vertex in a graph with maximum degree $d$. 
Therefore, given $v$, we can encode each connected induced subgraph containing $v$ with a positive integer $\varXi\in[d']$. 
In other words, there exists an injective mapping $\Upsilon_v$ from all connected induced subgraphs of $G$ containing $v$ to $\{v\}\times[d']$. 

Our counting argument will be based on encoding the whole $2$-block-tree. 
Intuitively, the encoding contains $\ell+1$ components.
The first one encodes how $C_i$'s are connected in $G^2$, 
and the rest encodes each individual $C_i$ by an integer in $[d']$. 

Let $\=T_{\theta d^2}$ to be the infinite $\theta d^2$-ary tree. 
In the first step, the relation between blocks is encoded by a subtree of $\=T_{\theta d^2}$ containing its root, which is basically a DFS tree. 
However, the order of visiting will affect the DFS tree we construct.
For this reason, we need to specify this ordering. 
First, we order the vertices by their indices. 
That is, $v_i\prec v_j$ if $i<j$. 
Given a subset $C$ of vertices, consider the set $\Gamma^2(C)$ containing vertices of distance $2$ from $C$. 
We can sort this set according to the ordering of vertices, 
and hence any vertex $u\in \Gamma^2(C)$ has a rank among $\Gamma^2(C)$, denoted by $\mathtt{Rank}_C(u)$. 
Suppose at some stage of our DFS algorithm, we have just finished handling some block $C$. 
Then we find the next unvisited vertex in $\Gamma^2(C)$, say $v'$, which is in some block $C'$ that needs to be encoded. 
Then $C'$ will be encoded as the $\mathtt{Rank}_C(v')$-th child of current vertex in the DFS tree, 
together with the integer $\Upsilon_{v'}(C')\in[d']$. 
The key of our proof is to show that this encoding is injective, 
i.e., no two distinct $2$-block-trees share the same encoding. 

With all the preparation, we give the encoding algorithm as \Cref{alg:encoding}. 
Once again, \Cref{alg:encoding} is for analysis only and does not need to be implemented.

\begin{algorithm}[ht] 
\caption{\textsf{Encoding}}
\label{alg:encoding}
\SetKwProg{Fn}{Procedure}{\string:}{}
\KwIn{A graph $G$, a vertex $v\in G$, a $2$-block-tree $\{C_1,\cdots,C_\ell\}$ of block size $\theta$ and tree size $\ell$}
\KwOut{An encoding $(T,\varXi_1,\cdots,\varXi_\ell)$, where $T$ is a subtree of $\mathbb{T}_{\theta d^2}$ of size $\ell$}
initialize \texttt{visited[1..$\ell$]} to be all \texttt{False}\;
let $C_j$ be the component containing $v$\;
let $r$ be the root of $\mathbb{T}_{\theta d^2}$\;
let $T$ be an empty subtree\;
$t\gets 0$\;
\textsf{DFS-Encode}($j$,$v$,$r$)\;
\Return $(T,\varXi_1,\cdots,\varXi_\ell)$\;
\BlankLine
\Fn{\textsf{DFS-Encode}($i$,$u$,$w$)}{
  \texttt{visited[i]} $\gets$ \textbf{True}\;
  $t\gets t+1$\;
  $\varXi_t\gets \Upsilon_u(C_i)$\;
  add $w$ into $T$\;
  \For(\tcp*[h]{enumerate $u'\in \Gamma^2(C_i)$ in order}){$u'\in \Gamma^2(C_i)$}{
    \If{there does not exist any $i'$ such that $C_{i'}\ni u'$\label{line:exist-check}}{\textbf{continue}\;}
    let $i'$ be the index such that $C_{i'}\ni u'$\;
    \If{\texttt{\emph{visited[i']}}$=$\texttt{\emph{False}}\label{line:visit-check}}{
      let $w'$ be the $\mathtt{Rank}_{C_i}(u')$-th child of $w$ in $\mathbb{T}_{\theta d^2}$\;\label{alg:encoding:nextw}
      \textsf{DFS-Encode}($i'$,$u'$,$w'$)\;
    }
  }

}
\end{algorithm}


\begin{lemma} \label{lemma-encoding-block-tree}
Fix a graph $G$ and a vertex $v$. Any $2$-block-tree $\{C_1,\cdots,C_\ell\}$ of block size $\theta$ and tree size $\ell$ containing $v$ can be encoded by a tuple $(T,\varXi_1,\cdots,\varXi_\ell)$, where $T$ is a subtree of $\=T_{\theta d^2}$ of size $\ell$ containing its root, and $\varXi_{i}\in[ d']$. Moreover, no two distinct $2$-block-trees share the same encoding. 
\end{lemma}

\begin{proof}
  The first part of this lemma follows by going through \Cref{alg:encoding}. 
  There are two things to verify: 
  \begin{itemize}
    \item The algorithm will always halt, outputting $\ell$ $\varXi_i$'s. 
      To show this, one only needs to check that every $C_i$ will be visited exactly once, 
      which is true due to property \ref{def-2-block-tree-3} of \Cref{def:2-block-tree} and \Cref{line:visit-check} of \Cref{alg:encoding}. 
    \item The algorithm can find such $w'$ on \Cref{alg:encoding:nextw}, or equivalently, $\mathsf{Rank}_{C_i}(u')\in [\theta d^2]$. 
      This follows after a trivial upper bound on the number of distance-$2$ neighbours that $|\Gamma^2(C_i)|\leq\theta d^2$. 
  \end{itemize}

To prove the second part, suppose there are two $2$-block-trees $\{C_1,\cdots,C_\ell\}$ and $\{C'_1,\cdots,C'_\ell\}$ with the same encoding $(T,\varXi_1,\cdots,\varXi_\ell)$. 
Without loss of generality, we can assume $C_1,\cdots,C_\ell$ (resp. $C'_1,\cdots,C'_\ell$) are sorted in the order of being visited by \Cref{alg:encoding}. 
The goal is then to prove $C_i=C'_i$ for all $i\in[\ell]$. 
To show this, we do a simple induction argument.
More precisely, denote by $T_t$ and $T'_t$ the subtrees constructed by the first $t$ calls to \textsf{DFS-Encode} respectively.
We induce on $t$ to show that
\begin{equation}\label{eqn:IH}
  C_i=C'_i \text{ for all }i\in[t], \text{ and } T_t=T'_t. \tag{IH}
\end{equation}

\paragraph*{Base case $t=1$.} Note that $C_1=C'_1$ follows from the injectivity of $\Upsilon_v$, and $T_1=T_1'$ as they both contain only the root. 

\paragraph*{Induction step.} Suppose \eqref{eqn:IH} holds for $t-1$. At this stage, we compare the progress of two copies of \textsf{Encoding} running on $C$ and $C'$ respectively. 
Right before the \textbf{for}-loop in the $(t-1)$-th call to \textsf{DFS-Encode}, 
both copies get the same $w$ by \eqref{eqn:IH}. 
Again by \eqref{eqn:IH}, both copies get the same $C_{t-1}$ in the condition of the \textbf{for}-loop. 
In the enumeration of \textbf{for}-loop, both copies skip or keep the $u'$ in \Cref{line:exist-check} simultaneously, because $C_i=C_i'$ for all $i\in [t-1]$. 
Note that each vertex of $\=T_{\theta d^2}$ can be visited at most once. 
This means that if the two copies get different $u'$ in \Cref{alg:encoding:nextw}, 
then the final subtree will be different. 
Therefore, they must get the same $u'$ and $i'$, and hence the same $w'$ because they have the same $C_{t-1}$, implying $T_{t}=T'_{t}$. 
Moreover, the next calls to \textsf{DFS-Encode} have an identical input in both copies.
Thus, $\varXi_t=\Upsilon_u(C_t)$ and $\varXi_t'=\Upsilon_u(C_t')$.
By assumption $\varXi_t=\varXi_t'$.
Injectivity of $\Upsilon_u$ implies that $C_t=C'_t$, finishing the proof. 
\end{proof}
	
We conclude this subsection by proving \Cref{lemma-block-tree-count}. 

\begin{proof}[Proof of \Cref{lemma-block-tree-count}]
By \Cref{lemma-encoding-block-tree}, the number of $2$-block-trees can be upperbounded by the number of possible encodings. 
To count the number of possible subtrees $T$, we simply apply \Cref{lem:connect_count}, which gives $(e\theta d^2)^{\ell-1}/2$. 
The number of possible $\varXi_i$ sequences is $d'^\ell=(ed)^{\ell(\theta-1)}/2^{\ell}$. 
Combining both parts yields the upper bound 
$\theta^{\ell-1}e^{\theta\ell-1}d^{(\theta+1)\ell-2}/2^{\ell+1}$.
\end{proof}

\section{Mixing of systematic scan}
\label{sec:mixing}
In this section, we prove the mixing lemma for the projected systematic scan Markov chain of hypergraph colourings (\Cref{lemma-mixing}).
First, we verify that the systematic scan is irreducible, aperiodic and reversible with respect to $\nu$.
This implies that the systematic scan has the unique stationary distribution $\nu$.
Aperiodicity and reversibility are straightforward to verify.
For irreducibility, it suffices to show that for any $\tau \in [s]^V$, $\nu(\tau) > 0$,
as our chain is a Glauber dynamics for $\nu$.
Fix an arbitrary configuration $\tau \in [s]^V$.
We show that there exists a proper colouring $\sigma \in [q]^V$ such that $h(\sigma) = \tau$, where $h$ is the projection scheme.
This implies $\nu(\tau) > 0$.
To prove the existence of such a proper colouring, consider the list hypergraph colouring instance $(H,(Q_v)_{v \in V})$, where $Q_v = h^{-1}(\tau_v)$ for all $v \in V$.
We only need to show that this list colouring instance has a feasible solution.
Note that $\abs{Q_v} \geq \ftp{q/\ctp{\sqrt{q}}}\geq \sqrt{q}/2$ for $q \geq 20$. 
By the Lov\'asz local lemma, \Cref{thm:lll}, we only need to verify that
\begin{align*}
\mathrm{e} q \tp{\frac{2}{\sqrt{q}}}^k\Delta k \leq 1,
\end{align*}
which follows from $q \geq 40 \Delta^{\frac{2}{k-4}}$ and $k \geq 20$.

Next, we prove the mixing time result in \Cref{lemma-mixing}.
The analysis is based on an information percolation argument.
We first define a coupling $\+C$ of the systematic scan $(\*X_t,\*Y_t)_{t\geq 0}$.
Let  $\*X_0,\*Y_0 \in [s]^V$ be two arbitrary initial configurations.
In the $t$-th transition step, 
\begin{itemize}
\item let $v \in V$ be the vertex with label $(t \mod n)$ and set $(X_t(u),Y_t(u)) \gets (X_{t-1}(u),Y_{t-1}(u ))$ for all other vertices $u \in V \setminus \{v\}$;
\item sample $(X_t(v),Y_t(v))$ from the  optimal coupling between $\nu^{X_{t-1}(V \setminus \{v\} )}_v$ and $\nu^{Y_{t-1}(V \setminus \{v\} )}_v$.
\end{itemize}
We prove the following lemma in this section.
\begin{lemma}\label{lemma-inf-per}
  Suppose $k \geq 20$ and $q \geq 40\Delta^{\frac{2}{k-4}}$.
For any initial configurations $\*X_0,\*Y_0 \in [s]^V$, any $\epsilon\in (0,1)$, let $T = \ctp{50 n \log \frac{n\Delta}{\epsilon}}$, it holds that
\begin{align*}
\forall v \in V, \quad \Pr_{\+C}\sqtp{X_T(v) \neq Y_T(v)}	\leq \frac{\epsilon}{n}.
\end{align*}
\end{lemma}
By \Cref{lemma-inf-per}, a union bound over all vertices and the coupling lemma (\Cref{lem:coupling_lemma}), it holds that
\begin{align*}
\max_{\*X_0,\*Y_0 \in [s]^V}\DTV{\*X_T}{\*Y_T} \leq \Pr_{\+C}\sqtp{\*X_T \neq \*Y_T} \leq \epsilon,
\end{align*}
which proves the mixing time part of \Cref{lemma-mixing} via \eqref{eqn:coupling-mixing}.
In the rest of this section, we use the information percolation technique to analyse the coupling $\+C$ and prove \Cref{lemma-inf-per}.

\subsection{Information percolation analysis}
Consider the coupling procedure $(\*X_t,\*Y_t)_{t \geq 0}$.
For each $t \geq 1$, let $v_t$ denote the vertex picked in the $t$-th step of systematic scan, namely, $v_t$ is the vertex with label $(t \mod n)$.
Consider the $t$-th transition step, where $t > 0$.
Define the set of agreement vertices when updating $v_t$ at time $t$ by
\begin{align*}
A_t \defeq \{v \in V \setminus \{v_t\} \mid X_{t-1}(v) = Y_{t-1}(v)\}.	
\end{align*}
We say a hyperedge $e \in \+E$ is satisfied by $A_t$ if there exist two distinct vertices $u,v \in e \cap A_t$ such that $X_{t-1}(u) \neq X_{t-1}(v)$ (and hence $Y_{t-1}(u) \neq Y_{t-1}(v)$ ).
We remove all the hyperedges  $e \in \+E$ satisfied by $A_t$ to obtain a sub-hypergraph $H_t$.
Let $H_t^v$ denote the connected component in $H_t$ containing $v$.
%
\begin{lemma}\label{lemma-v-find-u}
If $X_t(v_t) \neq Y_t(v_t)$ for some $t \geq 1$, then there exists $u\neq v_t$ in $H^{v_t}_t$ such that $X_{t-1}(u) \neq Y_{t-1}(u)$. 	
\end{lemma}
\begin{proof}
Note that  $X_t(v_t)$ and $Y_t(v_t)$ are sampled from 	$\nu_{v_t}^{X_{t-1}(V \setminus \{v_t\})}$ and $\nu_{v_t}^{Y_{t-1}(V \setminus \{v_t\})}$ respectively.
Let $\mu'$ denote the uniform distribution of proper colourings of $H^v_t$.
Let $\pi$ denote the projected distribution induced by $\mu'$ and the projection scheme $h$.
Let $V^{v_t}_t$ denote the vertex set of $H^{v_t}_t$ and let $S = V^{v_t}_t \setminus \{v_t\}$.
We claim that (1) $\nu_{v_t}^{X_{t-1}(V \setminus \{v_t\})}$ and $\pi^{X_{t-1}(S)}_{v_t}$ are identical distributions; (2) $\nu_{v_t}^{Y_{t-1}(V \setminus \{v_t\})}$ and $\pi^{Y_{t-1}(S)}_{v_t}$ are identical distributions.
Hence, if $X_{t-1}(u) = Y_{t-1}(u)$ for all $u\neq v_t$ in $H^{v_t}_t$ , then $X_t(v_t)$ and  $Y_t(v_t)$ must be perfectly coupled.

We verify that  $\nu_{v_t}^{X_{t-1}(V \setminus \{v_t\})}$ and $\pi^{X_{t-1}(S)}_{v_t}$ are identical distributions. The claim for  $\nu_{v_t}^{Y_{t-1}(V \setminus \{v_t\})}$ and $\pi^{Y_{t-1}(S)}_{v_t}$ can be verified by a similar proof. Consider the list colouring instance $(H,(Q_v)_{v \in V})$, where $Q_v = [q]$ if $v  = v_t$ and $Q_v = h^{-1}(X_{t-1}(v))$ if $v \neq v_t$. Let $\mu_{\mathrm{list}}$ denote the uniform distribution of all proper list colourings. 
If $X \sim \mu_{\mathrm{list}}$, then $h(X_{v_t}) \sim \nu_{v_t}^{X_{t-1}(V \setminus \{v_t\})}$.
For any hyperedge $e$ satisfied by $A_t$, it holds that for any colouring $X \in \otimes_{v\in V}Q_v$, $e$ is not monochromatic. 
Let $H_t$ denote the hypergraph obtained from $H$ by removing all hyperedges satisfied by $A_t$.
Hence, $(H,(Q_v)_{v \in V})$ and $(H_t, (Q_v)_{v \in V})$ have the same set of proper list colourings. 
Recall that $H^{v_t}_t$ is the connected component in $H_t$ containing vertex $v_t$.
Let $\mu_{\mathrm{list}}^{\mathrm{com}}$ denote the uniform distribution over all proper list colourings of $(H_t^{v_t},(Q_v)_{v \in V^{v_t}_t})$.
Hence, $\mu_{\mathrm{list}}$ projected on $v_t$ is the same distribution as  $\mu_{\mathrm{list}}^{\mathrm{com}}$ projected on $v_t$.
If $X \sim \mu_{\mathrm{list}}^{\mathrm{com}}$, then $h(X_{v_t}) \sim \pi^{X_{t-1}(S)}_{v_t}$.
This implies that  $\nu_{v_t}^{X_{t-1}(V \setminus \{v_t\})}$ and $\pi^{X_{t-1}(S)}_{v_t}$ are identical distributions. 
\end{proof}

We say that a hyperedge sequence $e_1,e_2,\ldots,e_\ell$ is a path in a hypergraph if for each $1< i\leq \ell$, $e_{i-1} \cap e_{i} \neq \emptyset$ and $e_{i -1 } \neq e_{i}$.
The following result is a straightforward corollary of \Cref{lemma-v-find-u}.
\begin{corollary}\label{corollary-v-find-u}
Let $t \geq 1$.
If $X_t(v_t) \neq Y_t(v_t)$, then there exists a vertex $u\neq v_t$ satisfying $X_{t-1}(u) \neq Y_{t-1} (u)$ and a path $e_1,e_2,\ldots,e_\ell$ in hypergraph $H$ such that
\begin{itemize}
\item $v \in e_1$ and $u \in e_\ell$;
\item for any hyperedge $e_i$ in the path, there exists $c \in [s]$ such that for all vertex $w \in e_i$ and $w\neq v_t$, either $X_{t-1}(w) = Y_{t-1}(w) = c$ or $X_{t-1}(w) \neq Y_{t-1}(w)$.	
\end{itemize}
\end{corollary}
\begin{proof}
By \Cref{lemma-v-find-u}, there is a vertex $u\neq v_t$ such that $X_{t-1}(u) \neq Y_{t-1}(u)$  and $u\in H^{v_t}_t$. 
As $u$ and $v_t$ are in the same connected component, there exist a path from $v_t$ to $u$.
%
%
%
Moreover, for each hyperedge $e_i$ on this path,
since $e_i$ is in $H^{v_t}_t$,
it is not satisfied by $A_t$.
This implies that for all $w\neq v_t\in e_i$ such that $X_{t-1}(w)=Y_{t-1}(w)$,
their values in both chains must be the same $c\in[s]$.
Lastly, note that any path in $H^{v_t}_t$ is also a path in $H$. 
This proves the corollary.
\end{proof}
\Cref{corollary-v-find-u} is a key result for the information percolation analysis.
For any time $0\leq t \leq T$, any vertex $v \in V$, define the set of previous update times by
\begin{align*}
S(v,t) \defeq \{1 \leq i \le t \mid v_i = v\},	
\end{align*}
where $v_i$ is the vertex picked in the $i$-th transition step. Define the last update time for $v$ up to $t$ by
\begin{align*}
\upt(v, t) \defeq \begin{cases}
  \max_{i\in S(v,t)} i &\text{if } S(v,t) \neq \emptyset;\\
 0	&\text{otherwise}.
 \end{cases}
 \end{align*}

By \Cref{corollary-v-find-u}, if the coupling on vertex $v$ failed at time $t$, then there must exist a vertex $u$ such that the coupling on $u$ failed at time $t' = \upt(u,t)$.
We apply  \Cref{corollary-v-find-u} recursively until we find a vertex $w$ such that $X_0(w) \neq Y_0(w)$.
This gives us an update time sequence $t = t_1 > t_2>\ldots>t_{\ell} = 0$ such that the coupling of each $t_i$-th transition fails, together with a set of paths satisfying the properties in \Cref{corollary-v-find-u}.
We will show that such a update time sequence and the set of paths occur with small probability, which bounds the probability of $X_t(v_t) \neq Y_t(v_t)$.
For this analysis, we will use the notions of extended hyperedges and extended hypergraphs introduced by He, Sun, and Wu \cite{HSW21}.

\subsection{Extended hyperedges and the extended hypergraph}
Fix an integer $T \geq 1$ to be the total number of transitions of the systematic scan.
Define the set of extended vertex $\ext{V}$ by
\begin{align*}
\ext{V} = \{(t,v_t) \mid 1\leq t \leq T\} \cup \{(0,v) \mid v \in V\},	
\end{align*}
where $v_t$ is the vertex with label $(t \mod n)$. Each vertex $(t,u) \in \ext{V}$ represents an update, i.e. $u$ is updated at the $t$-th transition step. 
We regard all vertices ``updated'' at the initial step ($t = 0$).
Consider the systematic scan process $(\*X_t)_{t \geq 0}$. 
For any hyperedge $e \in \+E$, the configuration $X_t(e)$ of $e$ at time $t$ satisfies 
\begin{align*}
  \forall u \in e, \quad X_t(u) = X_{t'}(u),\quad \text{where } t' = \upt(u,t),
\end{align*}
namely, the value of $u$ at time $t$ is the same as the value of $u$ at time $t' = \upt(u,t)$.
Besides, the configuration of hyperedge $e$ remains unchanged until some vertex in $e$ is updated. 
This motivates the following definition of extended hyperedges and the extended hypergraph, introduced by He, Sun, and Wu \cite{HSW21}.
\begin{definition}\label{definition-extended}
  The set $\ext{\+E}$ of \emph{extended hyperedges} is defined by $\ext{\+E} \defeq \cup_{t=0}^T \ext{\+E}_t$, where
  \begin{align*}
    \ext{\+E}_0 &\defeq \bigcup_{e \in \+E} \{ (0,v) \mid v \in e \},\\
    \forall 1\leq t \leq T,\quad \ext{\+E}_t &\defeq \bigcup_{e:v_t \in e}\left\{ (t',v) \mid v \in e \land t' = \upt(v,t) \right\}.	
  \end{align*}
  The \emph{extended hypergraph} is $\ext{H}=(\ext{V},\ext{\+E})$.
\end{definition}

At the beginning, each hyperedge $e \in \+E$ takes its initial value, and thus we add all the extended hyperedges with $t=0$ to $\ext{\+E}_0$. 
For each update at time $1 \leq t \leq T$, only the value of $v_t$ is updated.
Thus the configurations of only the hyperedges containing $v_t$ are updated, and we add only those to $\ext{\+E}_t$. 

\Cref{corollary-v-find-u} shows that for any $t \geq 1$, if the coupling in the $t$-th transition step fails (i.e. $X_t(v_t) \neq Y_t(v_t)$), then we can find a specific path in the hypergraph $H$. 
Our next lemma 
lifts such a path to $\ext{H}$.

\begin{lemma}\label{lemma-v-find-u-extend}
Let $1\leq t \leq T$ be an integer. Suppose $X_t(v_t) \neq Y_t(v_t)$. There exist a vertex $(t',u) \in \ext{V}$ satisfying $t' < t$ and $X_{t'}(u) \neq Y_{t'}(u)$, together with a path $\ext{e}_1,\ext{e}_2,\ldots,\ext{e}_\ell$ in $\ext{H}$ such that
\begin{itemize}
\item $(t,v_t) \in \ext{e}_1$ and $(t',u) \in \ext{e}_\ell$;
\item for any hyperedge $\ext{e}_i$ in the path, there exists $c \in [s]$ such that for all $(j,w) \in \ext{e}_i$, either $X_j(w)=Y_j(w) = c$ or $X_j(w) \neq Y_j(w)$.
\end{itemize}
\end{lemma}
\begin{proof}
Let $u$ and $e_1,e_2,\ldots,e_{\ell}$ denote the vertex and the path in \Cref{corollary-v-find-u} respectively.
For each vertex $w \in V$, let $t_w = \upt(w,t)$. 
For each $1\leq i \leq \ell$, define 
\begin{align*}
\ext{e}_i = \{(t_w,w)\mid w \in e_i\}.	
\end{align*}
To show that $\ext{e}_1,\ext{e}_2,\ldots,\ext{e}_\ell$ is a path in $\ext{H}$, we need to verify that each $\ext{e}_i$ defined above belongs to $\ext{\+E}$ in \Cref{definition-extended}. Fix an $\ext{e}_i$. Let $t_{\max} = \max\{t \mid (t,w) \in \ext{e}_i\}$.
It is straightforward to verify that $\ext{e}_i \in \ext{\+E}_{t_{\max}}$.

Next, we show that $t'< t$ and $X_{t'}(u) \neq Y_{t'}(u)$. 
By definition, we have $t' = t_u  = \upt(u,t) < t$.
As the value of any vertex does not change until the next update, we have that
\begin{align}
\label{eq-t-1-to-last}
\forall w \in V \setminus \{v_t\}, \quad X_{t-1}(w) = X_{t_w}(w)	 \text{ and } Y_{t-1}(w) = Y_{t_w}(w).
\end{align}
By \Cref{corollary-v-find-u}, it holds that $X_{t-1}(u) \neq Y_{t-1}(u)$.
By~\eqref{eq-t-1-to-last}, it holds that $X_{t'}(u) \neq Y_{t'}(u)$.

Finally, we verify the two properties of the path. The first property  $(t,v_t) \in \ext{e}_1$ and $(t',u) \in \ext{e}_\ell$ follows from the way $\ext{e}_i$ is constructed. 
By \Cref{corollary-v-find-u}, for any $e_i$ in the path, there exists $c \in [s]$ such that for all vertices $w \in e_i \setminus \{v_t\}$, either $X_{t-1}(w) = Y_{t-1}(w) = c$ or $X_{t-1}(w) \neq Y_{t-1}(w)$.
By~\eqref{eq-t-1-to-last}, for all extended vertices  $(i,w) \in \ext{e}_i$ with $w \neq v_t$, either $X_{i}(w) = Y_{i}(w) = c$ or $X_{i}(w) \neq Y_{i}(w)$.
Finally, consider the extended vertex $(t, v_t)$. By our assumption in the lemma, we have that $X_t(v_t) \neq Y_t(v_t)$. 
%
\end{proof}

We may repeatedly apply \Cref{lemma-v-find-u-extend} to trace a discrepancy from some time $t$ to time $0$.

\begin{lemma}\label{lemma-path}
Let $1\leq t \leq T$ be an integer. Suppose $X_t(v_t) \neq Y_t(v_t)$. There exists a path $\ext{e}_1,\ext{e}_2,\ldots,\ext{e}_{\ell}$ in the extended hypergraph $\ext{H}$ such that
\begin{itemize}
\item  $(t,v_t) \in \ext{e}_1$, $\min\{j \mid (j,w) \in \ext{e}_{i} \} > 0$ for all $i < \ell$ and $\min\{j \mid (j,w) \in \ext{e}_{\ell} \} = 0$;
\item for any $1\leq i,i' \leq \ell$ satisfying $\abs{i-i'} \geq 2$, $\ext{e}_{i} \cap \ext{e}_{i'} = \emptyset$;
\item  for any hyperedge $\ext{e}_i$ in the path, there exists $c \in [s]$ such that for all $(j,w) \in \ext{e}_i$, either $X_j(w)=Y_j(w) = c$ or $X_j(w) \neq Y_j(w)$.
\end{itemize}
\end{lemma}
\begin{proof}
  We use \Cref{lemma-v-find-u-extend} recursively. Namely, we use \Cref{lemma-v-find-u-extend} for $(t,v_t)$ to find $(t',u)$.
If $t' \neq 0$, we apply \Cref{lemma-v-find-u-extend} on $(t',u)$ again to find the previous discrepancy.
Repeat this process until we find $(t'',w)$ such that $t'' = 0$.
This gives a path $\ext{f}_1,\ext{f}_2,\ldots,\ext{f}_{m}$ in the extended hypergraph $\ext{H}$ such that $(t,v_t) \in \ext{f}_1$ and $\min\{j \mid (j,w) \in \ext{f}_{m} \} = 0$. 
By \Cref{lemma-v-find-u-extend}, this path  $\ext{f}_1,\ext{f}_2,\ldots,\ext{f}_{m}$ satisfies the last property in \Cref{lemma-path}.
%

We then construct the path $\ext{e}_1,\ext{e}_2,\ldots,\ext{e}_{\ell}$.
First let $\ext{e}_1 = \ext{f}_1$, $\ell = 1$, and $p = 1$.
While $\min\{i \mid (i,w) \in \ext{e}_\ell \} > 0$, we repeat the following process:
\begin{itemize}
\item let $p+1 \leq j \leq m$ be the largest index satisfying $\ext{f}_j \cap \ext{e}_{\ell} \neq \emptyset$;
\item let $\ell \gets \ell + 1$, $\ext{e}_\ell \gets \ext{f}_j$ and $p \gets j$.
\end{itemize}
When the above process ends, we get the path $\ext{e}_1,\ext{e}_2,\ldots,\ext{e}_{\ell}$.

We first show that the process above is well-defined.
Consider the beginning of each iteration of the while-loop. It holds that $\ext{e}_\ell = \ext{f}_p$. Since  $\min\{i \mid (i,w) \in \ext{e}_\ell \} > 0$, we know that $p < m$. The index $p + 1\leq j \leq m$ such that  $\ext{f}_j \cap \ext{e}_{\ell} \neq \emptyset$ must exist because $\ext{f}_{p+1} \cap \ext{e}_{\ell} =\ext{f}_{p+1} \cap \ext{f}_{p} \neq \emptyset$. 
The while-loop must terminate eventually because $p$ always increase and $\min\{i \mid (i,w) \in \ext{f}_m \} = 0$.

We claim that $\ext{e}_1,\ext{e}_2,\ldots,\ext{e}_{\ell}$ is indeed a path. We only need to show that for all $2\leq i \leq \ell$, it holds that $\ext{e}_{i-1} \cap \ext{e}_i \neq \emptyset$ and $\ext{e}_{i-1} \neq \ext{e}_i$. The construction process guarantees that $\ext{e}_{i-1} \cap \ext{e}_i \neq \emptyset$. Suppose there is an index $2\leq i \leq \ell$ such that $\ext{e}_{i-1} = \ext{e}_i = \ext{f}_{i'}$ for some $i'\le m$. Since the construction process finds $\ext{e}_i$, we know that $\min\{t \mid (t,w) \in \ext{e}_{i-1} \} > 0$. Thus $i' < m$ and $\ext{f}_{i'+1}$ exists. Since $\ext{f}_1,\ext{f}_2,\ldots,\ext{f}_{m}$ is a path, we know that $\ext{f}_{i'} \cap \ext{f}_{i'+1} \neq \emptyset$, which implies that $\ext{e}_{i-1} \cap \ext{f}_{i' + 1} \neq \emptyset$. When constructing $\ext{e}_{i}$, we look for the largest $j$ such that $\ext{e}_{i-1} \cap \ext{f}_j \neq \emptyset$. Hence, $\ext{e}_i \neq \ext{f}_{i'}$, a contradiction.

Lastly, we verify the properties of the path.
\begin{itemize}
  \item Since $\ext{e}_1 = \ext{f}_1$ and $(t,v_t) \in \ext{f}_1$, $(t,v_t) \in \ext{e}_1$.
    The while-loop terminates once $\min\{j \mid (j,w) \in \ext{e}_\ell \} > 0$. Hence, $\min\{j \mid (j,w) \in \ext{e}_{i} \} > 0$ for all $i < \ell$ and $\min\{j \mid (j,w) \in \ext{e}_{\ell} \} = 0$.
  \item For any $1\leq i,i'\leq \ell$ with $\abs{i-i'} \geq 2$,
    consider how $\ext{e}_{i+1}$ is constructed.
    We choose the largest index $j\le m$ such that $\ext{f}_j \cap \ext{e}_{\ell} \neq \emptyset$ and $\ext{e}_{i+1}\gets \ext{f}_j$.
    In other words, for any $j'>j$, $\ext{f}_{j'} \cap \ext{e}_{i} = \emptyset$.
    Since there is $j'$ such that $\ext{e}_{i'}=\ext{f}_{j'}$,
    $\ext{e}_{i} \cap \ext{e}_{i'} = \emptyset$.
%
  \item Since $\ext{e}_1,\ext{e}_2,\ldots,\ext{e}_{\ell}$ is a subsequence of $\ext{f}_1,\ext{f}_2,\ldots,\ext{f}_{m}$, the last property is satisfied as well. \qedhere
\end{itemize}
%
%
\end{proof}

\subsection{Proof of \texorpdfstring{\Cref{lemma-inf-per}}{Lemma \ref{lemma-inf-per}}}
Recall that $T = \ctp{50 n \log \frac{n}{\epsilon}}$ in \Cref{lemma-inf-per}.
To prove \Cref{lemma-inf-per}, we need to show that 
\begin{align*}
\forall v \in V, \quad \Pr_{\+C}\sqtp{X_T(v) \neq Y_T(v)}	\leq \frac{\epsilon}{n}.
\end{align*}
Fix a vertex $v$. By the same reason as \eqref{eq-t-1-to-last}, we only need to prove $\Pr_{\+C}\sqtp{X_T(v) \neq Y_T(v)}	\leq \frac{\epsilon}{n}$ for a new $T$, where 
\begin{align}\label{eq-def-new-T}
  T = \upt\tp{v,\ctp{50 n \log \frac{n}{\epsilon}}} \geq \ctp{40 n \log \frac{n}{\epsilon}}.
\end{align}
Note that $v$ is updated at time $T$, i.e.~$v = v_T$.

Fix $T$ defined in~\eqref{eq-def-new-T}. Define the following information percolation path (IPP).

\begin{definition}\label{definition-IPP}
We say a path $\ext{e}_1,\ext{e}_2,\ldots,\ext{e}_{\ell}$ of length $\ell$ in the extended hypergraph $\ext{H}$ is an \emph{information percolation path} (\emph{IPP}) if the following two properties are satisfied:
\begin{itemize}
\item $(T,v_T) \in \ext{e}_1$, $\min\{j \mid (j,w) \in \ext{e}_{i} \} > 0$ for all $i < \ell$ and $\min\{j \mid (j,w) \in \ext{e}_{\ell} \} = 0$;
\item for any $1\leq i,j \leq \ell$ such that $\abs{i-j} \geq 2$, $\ext{e}_{i} \cap \ext{e}_{j} = \emptyset$.
\end{itemize}
\end{definition}

Suppose $X_T(v) \neq Y_T(v)$. By \Cref{lemma-path}, we can find an IPP $\ext{e}_1,\ext{e}_2,\ldots,\ext{e}_{\ell}$ in extended hypergraph $\ext{H}$.
The following lemma lower bounds the length of the IPP.
\begin{lemma}\label{lemma-length-lower-bound}
 For any IPP of length $\ell$, $\ell \geq \ctp{T/n}$.	
\end{lemma}
\begin{proof}
For any extended hyperedge $\ext{e}_i$, define the maximum and minimum update times in $\ext{e}_i$ by  $t^{(i)}_{\max} = \max\{t \mid (t,w) \in \ext{e}_i \}$	and $t^{(i)}_{\min} = \min\{t \mid (t,w) \in \ext{e}_i \}$. 
In the systematic scan, we update vertices in order of their labels.
By \Cref{definition-extended}, it holds that for any $i$,
\begin{align*}
t^{(i)}_{\max} - t^{(i)}_{\min} \leq n - 1 \leq n.	
\end{align*}
Note that $\ext{e}_i \cap \ext{e}_{i+1} \neq \emptyset$, which implies 
\begin{align*}
t^{(i)}_{\min} \leq t^{(i+1)}_{\max}	 \leq t^{(i+1)}_{\min} + n.
\end{align*}
Note that $t^{(1)}_{\min} \geq t^{(1)}_{\max} - n = T - n$.
We have
\begin{align*}
T - n \leq t^{(1)}_{\min} \leq t^{(\ell)}_{\min} + (\ell - 1) n = (\ell - 1)n,	
\end{align*}
where the last equation holds because $t^{(\ell)}_{\min} = 0$.
Since $\ell$ is an integer, we have $\ell \geq \ctp{T / n}$.
\end{proof}

Now fix  an integer $\ell \geq T /n$ and an IPP $\+P = \ext{e}_1,\ext{e}_2,\ldots,\ext{e}_\ell$ of length $\ell$. 
We define the bad event $\+B(\+P)$ as: for any hyperedge $\ext{e}_i$ in the path, there exists $c \in [s]$ such that for all $(j,w) \in \ext{e}_i$, either $X_j(w)=Y_j(w) = c$ or $X_j(w) \neq Y_j(w)$.
Namely, $\+B(\+P)$ that implies $\+P$ satisfies the third property in \Cref{lemma-path}.
By \Cref{lemma-path}, \Cref{lemma-length-lower-bound} and a union bound over all IPPs of length at least $\ell$, the probability of $X_T(v) \neq Y_T(v)$ can be bounded as follows 
\begin{align}\label{eq-union-bound-coupling}
\Pr_{\+C}\sqtp{X_T(v) \neq Y_T(v)} \leq \sum_{\ell \geq \ctp{T / n}}\sum_{\text{$\+P$: IPP of length $\ell$}} \Pr_{\+C}\sqtp{\+B(\+P)}.
\end{align}

We bound $ \Pr_{\+C}\sqtp{\+B(\+P)}$ in the RHS of~\eqref{eq-union-bound-coupling} next.
We need to use more delicate structures of the extended hypergraph $\ext{H} = (\ext{V},\ext{\+E})$. By \Cref{definition-extended}, each extended hyperedge $\ext{e} \in \ext{\+E}$ corresponds to a unique hyperedge $\edge{\ext{e}} \in \+E$ in the input hypergraph, or more formally,
\begin{align*}
\edge{\ext{e}} \defeq \{ v \mid (t,v) \in \ext{e}\}.	
\end{align*}
We remark that different extended hyperedges may correspond to the same hyperedge.
For each extended hyperedge $\ext{e} \in \ext{\+E}$, we use $N(\ext{e})$ to denote the neighbour  extended hyperedges:
\begin{align*}
N(\ext{e}) \defeq \{\ext{f} \in \ext{\+E} \mid \ext{f} \cap \ext{e} \neq \emptyset \text{ and } \ext{f} \neq \ext{e}\}.	
\end{align*}
The following observation is straightforward to verify.
\begin{observation}\label{observation-neighbor-1}
For any $\ext{e} \in \ext{\+E}$ and $\ext{f} \in N(\ext{e})$, $\edge{\ext{e}} \cap \edge{\ext{f}} \neq \emptyset$.	
\end{observation}
\noindent
We further partition $N(\ext{e})$ into self-neighbours and outside-neighbours as follows, 
\begin{align*}
N_{\mathsf{self}}(\ext{e}) &\defeq \left\{\ext{f} \in N(\ext{e}) \mid \edge{\ext{e}} = \edge{\ext{f}} \right\};\\
N_{\mathsf{out}}(\ext{e}) &\defeq \left\{\ext{f} \in N(\ext{e}) \mid \edge{\ext{e}} \neq \edge{\ext{f}} \right\}.	
\end{align*}
\begin{observation}\label{observation-neighbor-2}
For any $\ext{e} \in \ext{\+E}$ and $\ext{f} \in N_{\mathsf{out}}(\ext{e})$, $\abs{\ext{e}  \cap \ext{f}} = 1$.	
\end{observation}
\begin{proof}
  Let $e = \edge{\ext{e}}$ and $f = \edge{\ext{f}}$. Since $\ext{f} \in N_{\mathsf{out}}(\ext{e})$, by \Cref{observation-neighbor-1} and the fact that the input hypergraph is simple, $\abs{e \cap f} = 1$, which implies $\abs{\ext{e}  \cap \ext{f}} = 1$.	
\end{proof}
The following lemma bounds the degree of the extended hypergraph.
\begin{lemma}\label{lemma-degree}
  Let $\Delta$ be the maximum degree of the input hypergraph $H = (V,\+E)$.
  Then,
  \begin{enumerate}
    \item given $(t,v)\in \ext{V}$ and $e\in\+E$ such that $v\in e$, the number of $\ext{e}$ such that $(t,v)\in \ext{e}$ and $\edge{\ext{e}}=e$ is at most $k$; \label{item:extend-case-1}
    \item for any extended vertex $(t,v) \in \ext{V}$, the number of extended hyperedges incident to $(v,t)$ is at most $d_{\mathrm{vtx}} \defeq \Delta k$; \label{item:extend-case-2}
    \item for any extended hyperedge $\ext{e} \in \ext{\+E}$, $N_{\mathsf{self}}(\ext{e})\le d_{\mathsf{self}} \defeq 2k$, 
      $N_{\mathsf{out}}(\ext{e})\le d_{\mathsf{out}} \defeq \Delta k^2$. \label{item:extend-case-3}
  \end{enumerate}
\end{lemma}
\begin{proof}
  For Item \ref{item:extend-case-1}, 
  suppose such $\ext{e}$ is $\{(t_j, u_j) \mid 1\leq j \leq k \}$ and $t_1 \leq t_2 \leq \ldots \leq t_k$. 
  Moreover, for all $j$ such that $t_j=0$, we order $u_j$ according to their original label in $H$.
  As $(v,t)\in \ext{e}$, $t$ equals one of $t_j$ .
  Then observe that $\ext{e}$ is uniquely determined if we know $t=t_j$ for some $1\le j\le k$, 
  and there are at most $k$ choices of $j$ (the number of choices can be less than $k$ if $t=0$).
  This shows the claim.

  For Item \ref{item:extend-case-2}, 
  if $\ext{e}$ is incident to $(v,t)$, then $\edge{\ext{e}} = e$ for some $e\ni v$. 
  There are at most $\Delta$ choices of such hyperedge $e$ in $H$. 
  Then the bound follows from Item \ref{item:extend-case-2}.

  For Item \ref{item:extend-case-3}, let $e = \edge{\ext{e}}$,
  and again assume $\ext{e}$ is $\{(t_j, u_j) \mid 1\leq j \leq k \}$ and $t_1 \leq t_2 \leq \ldots \leq t_k$ as in the proof of Item \ref{item:extend-case-1}.

  To bound the number of self-neighbours, suppose $\ext{f} \in N_{\mathsf{self}}(\ext{e})$ such that $\edge{\ext{f}} = e$. Let $t_{\max} = \max\{t \mid (t,w) \in \ext{f} \}$ and $t_{\min} = \min\{t \mid (t,w) \in \ext{f} \}$.
  Note that if $t_{\max} \leq t_k$, then there are at most $k-1$ choices of $t_{\max}$, namely $t_1,t_2,\ldots,t_{k-1}$.
  Otherwise $t_{\max} > t_k$. 
%
%
  Note that if $t_{\max} \geq  t_k + n $, then $t_{\min} \geq t_{\max} - (n - 1) > t_k$, which contradicts to $\ext{e}\cap\ext{f}\neq\emptyset$.
  It must hold that $t_k + 1 \leq t_{\max} \leq t_k + n - 1$. 
  In the interval $[t_k+1,t_k+n-1]$,
  there are at most $k-1$ times so that one of the vertices in $e$ is updated (this vertex cannot be $t_k$ as its update times are $t_k$ and $t_k+n$).
  Thus, there are $k-1$ choices of $t_{\max}$ again.
  Once $t_{\max}$ is fixed, since $\edge{\ext{f}} = e$, $\ext{f}$ is also fixed.
  Overall, the number of $\ext{f} \in N_{\mathsf{self}}(\ext{e})$ is at most $2(k-1) \leq 2k$.

To bound the number of outside-neighbours. We first choose one of the $k$ extended vertices in $\ext{e}$, say $(t_i,u_i)$. Then consider $\ext{f} \in N_{\mathsf{out}}(\ext{e})$ such that $(t_i,u_i) \in \ext{f}$. By Item \ref{item:extend-case-2}, the number of such $\ext{f}$ is at most $\Delta k$, implying the overall bound of $\Delta k^2$.
\end{proof}

Consider the IPP $\+P = \ext{e}_1,\ext{e}_2,\ldots,\ext{e}_\ell$. 
Define the parameters $R_{\mathrm{self}}$ and $R_{\mathrm{out}}$ by
\begin{align*}
R_{\mathrm{self}} &\defeq \abs{\left\{ 2\leq i \leq \ell \mid \ext{e}_{i} \in N_{\mathsf{self}}(\ext{e}_{i-1}) \right\} };\\
R_{\mathrm{out}} &\defeq \abs{\left\{ 2\leq i \leq \ell \mid \ext{e}_{i} \in N_{\mathsf{out}}(\ext{e}_{i-1}) \right\} }.
\end{align*}
By definition, $R_{\mathrm{self}}$ counts the number of consecutive self neighbours in $\+P$ and $R_{\mathrm{out}}$ counts the number of consecutive outside neighbours in $\+P$.
It holds that $R_{\mathrm{self}} + R_{\mathrm{out}} = \ell - 1$.
We have the following lemma.
\begin{lemma}\label{lemma-single-path}
Suppose $k \geq 20$ and $q \geq 40\Delta^{\frac{2}{k-4}}$.
For any IPP $\+P = \ext{e}_1,\ext{e}_2,\ldots,\ext{e}_\ell$, it holds that 
\begin{align*}
\Pr_{\+C}\sqtp{\+B(\+P)} \leq 10^3\Delta k^6\tp{\frac{1}{10^3\Delta k^6}}^{R_{\mathrm{out}} + \frac{1}{3}\tp{R_{\mathrm{self}} - b}},
\end{align*}
where $b$ is an integer satisfying $0\leq b \leq \min\{R_{\mathrm{self}}, 2R_{\mathrm{out}}\}$.
\end{lemma}

The proof of \Cref{lemma-single-path} is given in \Cref{section-missing-proof-mixing}, where we will specify the  value of the integer $b$.
Now, we use \Cref{lemma-single-path} to prove \Cref{lemma-inf-per}.
We remark that in the proof of \Cref{lemma-inf-per}, we do not use the specific value of $b$, we only use the fact that $0\leq b \leq \min\{R_{\mathrm{self}}, 2R_{\mathrm{out}}\}$. 

\begin{proof}[Proof of \Cref{lemma-inf-per}]
First fix an integer $\ell \geq \lfloor T / n \rfloor$ and an integer $0 \leq r \leq \ell-1$.
Consider the IPP $\+P$ of length $\ell$ such that $R_{\mathrm{out}} = r$ and $R_{\mathrm{self}}= \ell - 1 -r$. 
By the definition of IPP (\Cref{definition-IPP}) together with \Cref{lemma-degree}, the number of such path $\+P$ is at most 
\begin{align*}
\binom{\ell - 1}{r}d_{\mathrm{vtx}}d_{\mathrm{out}}^rd_{\mathrm{self}}^{\ell - 1 -r}	\leq \Delta k\binom{\ell - 1}{r}\tp{\Delta k^2}^{r} \tp{2k}^{\ell - 1- r}.
\end{align*}
By \Cref{lemma-single-path} and the union bound in~\eqref{eq-union-bound-coupling}, we have
\begin{align*}
&\Pr_{\+C}\sqtp{X_T(v) \neq Y_T(v)} \leq \sum_{\ell \geq \ctp{T / n}}\sum_{\text{$\+P$: IPP of length $\ell$}} \Pr_{\+C}\sqtp{\+B(\+P)}\\
\leq\,& \sum_{\ell \geq \ctp{T / n}}\sum_{r = 0}^{\ell - 1}\Delta k\binom{\ell - 1}{r}\tp{\Delta k^2}^{r} \tp{2k}^{\ell - 1- r} \cdot 10^3\Delta k^6\tp{\frac{1}{10^3\Delta k^6}}^{r+ \frac{1}{3}\tp{\ell - 1 -r - b(\ell,r)}},
\end{align*}
where $b(\ell,r)$ is an integer satisfying $0 \leq b(\ell,r) \leq  \min\{\ell - 1 -r,2r\}$. Since $b(\ell,r) \leq \ell - 1 -r$, it holds that $\tp{\frac{1}{10^3\Delta k^6}}^{\frac{\ell - 1 -r - b(\ell,r)}{3}} \leq \tp{\frac{1}{10^3 k^6}}^{\frac{\ell - 1 -r - b(\ell,r)}{3}}$, which implies
\begin{align*}
\Pr_{\+C}\sqtp{X_T(v) \neq Y_T(v)}&\leq \sum_{\ell \geq \ctp{T / n}}\sum_{r = 0}^{\ell - 1}\Delta k\binom{\ell - 1}{r}\tp{\Delta k^2}^{r} \tp{2k}^{\ell - 1- r} \cdot 10^3\Delta k^6\tp{\frac{1}{10^3\Delta k^6}}^{r} \tp{\frac{1}{10^3k^6}}^{\frac{\ell - 1 -r - b(\ell,r)}{3}}\\
&=10^3\Delta^2k^7\sum_{\ell \geq \ctp{T / n}}\sum_{r = 0}^{\ell - 1}\binom{\ell - 1}{r}\tp{\frac{1}{5k}}^{\ell - 1- r} \tp{\frac{1}{10^3k^4}}^{r} \tp{\frac{1}{10k^2}}^{- b(\ell,r)}.
\end{align*}
Note that $k \geq 20$.
Since $0\leq b(\ell,r) \leq 2r$, we have $\tp{\frac{1}{10k^2}}^{- b(\ell,r)} \leq \tp{\frac{1}{10k^2}}^{-2r} = (100k^4)^r$, which imples
\begin{align*}
\Pr_{\+C}\sqtp{X_T(v) \neq Y_T(v)} &\leq	10^3\Delta^2k^7\sum_{\ell \geq \ctp{T / n}}\tp{\frac{1}{10}}^{\ell - 1}\sum_{r = 0}^{\ell - 1}\binom{\ell - 1}{r} = 10^3\Delta^2k^7\sum_{\ell \geq \ctp{T / n}}\tp{\frac{1}{5}}^{\ell - 1}\\
&\leq 10^3\Delta^2k^7\tp{\frac{1}{2}}^{T/n}.
\end{align*}
Note that $T \geq 40 n \log \frac{n \Delta}{\epsilon}$ and $k \leq n$. We have
\begin{align*}
\Pr_{\+C}\sqtp{X_T(v) \neq Y_T(v)} \leq \frac{\epsilon}{n}. &\qedhere
\end{align*}
\end{proof}

\subsection{Proof of \texorpdfstring{\Cref{lemma-single-path}}{Lemma \ref{lemma-single-path}}}\label{section-missing-proof-mixing}
Fix an IPP $\+P = \ext{e}_1,\ext{e}_2,\ldots,\ext{e}_\ell$.
We define a total ordering among all extended hyperedges in $\+P$.
For any two extended hyperedges $ \ext{e}_i$ and $\ext{e}_j$ in $\+P$, we say $ \ext{e}_i < \ext{e}_j$ if and only if $i<j$.
%
\begin{lemma}\label{lemma-ind-path}
There exists a subsequence 	$\ext{f}_1 < \ext{f}_2 <\ldots <  \ext{f}_m$ in IPP $\+P$ such that
\begin{itemize}
\item for any $1\leq i,j \leq m$ satisfying $\abs{i-j} \geq 2$, $\ext{f}_i \cap \ext{f}_j = \emptyset$;
\item for any $2\leq i \leq m$, $\abs{\ext{f}_i \cap \ext{f}_{i-1}} \leq 1$;
\item $m \geq R_{\mathrm{out}} + \frac{1}{3}\tp{R_{\mathrm{self}} - b}$ for some integer $0\leq b \leq \min\{R_{\mathrm{self}}, 2R_{\mathrm{out}}\}$.	
\end{itemize}
\end{lemma}

Note that $\{\ext{f}_i\}$ given in \Cref{lemma-ind-path} is not necessarily a path.
What we do in \Cref{lemma-ind-path} is to prune certain self-neighbours from $\+P$ so that the second property holds.
To be more precise, for a maximal sequence of consecutive self-neighbouring hyperedges,
we prune all hyperedges that are in even positions of this sequence.
We give a formal proof below.

\begin{proof}[Proof of \Cref{lemma-ind-path}]
There are $\ell - 1$ pairs of adjacent extended hyperedges, i.e.~$\ext{e}_{i-1}$ and $\ext{e}_{i}$ are adjacent for $2\leq i \leq \ell$.
%
Define 
\begin{align*}
S_{\mathrm{out}} \defeq \left\{ \text{integer }i \in [2,\ell] \mid \ext{e}_{i} \in N_{\mathrm{out}}(\ext{e}_{i-1})\right\}.	
\end{align*}
Note that $\abs{S_{\mathrm{out}} } = R_{\mathrm{out}}$.
Denote $R = R_{\mathrm{out}}$.
Suppose the elements in $S_{\mathrm{out}}$ are $2\leq i_1 < i_2 <\ldots<i_{R} \leq \ell$.
In addition, we define $i_0 = 1$ and $i_{R + 1} = \ell+1$, although $i_0 \notin S_{\mathrm{out}}$ and $i_{R+1} \notin S_{\mathrm{out}}$.
Removing all the elements in $S_{\mathrm{out}}$, the integers in the interval $[2,\ell]$ splits into a set $I_{\mathrm{self}}$ of sub-intervals:
\begin{align*}
  I_{\mathrm{self}} \defeq \{[l,r] \mid \exists j \text{ s.t.~} 0\leq j \leq R,~l = i_j + 1,~r = i_{j+1} - 1, \text{ and }l \leq r\}.	
\end{align*}
Equivalently, $I_{\mathrm{self}}$ can be constructed by going through all $j$ from $0$ to $R$, and adding the interval $[i_j + 1,i_{j+1} - 1]$ to the set  $I_{\mathrm{self}}$ if $i_j + 1 \leq i_{j+1} - 1$.
For each interval $[l,r] \in I_{\mathrm{self}}$, the following properties hold 
\begin{enumerate}
  \item for each integer $i \in [l,r]$, $\ext{e}_{i} \in N_{\mathrm{self}}\tp{\ext{e}_{i-1}}$;\label{item:e-1}
  \item either $l = 2$ or $\ext{e}_{l-1} \in N_{\mathrm{out}}\tp{\ext{e}_{l-2}}$;\label{item:e-2}
  \item either $r = \ell$ or $\ext{e}_{r+1} \in N_{\mathrm{out}}\tp{\ext{e}_{r}}$.\label{item:e-3}
\end{enumerate}
In other words, each interval $[l,r] \in I_{\mathrm{self}}$ represents a sequence of consecutive extended hyperedges in the IPP $\+P$ of length $r-l+1$ such that each extended  hyperedge is a self-neighbour of its predecessor in $\+P$, and this sequence is maximal. 

Suppose the intervals in $I_{\mathrm{self}}$ are $[l_1,r_1],[l_2,r_2],\ldots,[l_a,r_a]$ such that $l_1 \leq r_1 < l_2 \leq r_2 < \ldots < l_a \leq r_a$, where $a = \abs{I_{\mathrm{self}}}$. 
It is straightforward to verify that 
\begin{align}\label{eq-total-length}
\sum_{i = 1}^a \tp{r_i - l_i + 1} = R_{\mathrm{self}}. 	
\end{align}

Define a subset $I^{(1)}_{\mathrm{self}}\subseteq I_{\mathrm{self}}$ by
\begin{align*}
I^{(1)}_{\mathrm{self}} \defeq \left\{ [l,r] \in I_{\mathrm{self}} \mid l = r \right\}.	
\end{align*}
The quantity $b$ is the size of $I^{(1)}_{\mathrm{self}}$, i.e.~$b \defeq \abs{I^{(1)}_{\mathrm{self}}}$.
Since $I^{(1)}_{\mathrm{self}}$ is a subset of $I_{\mathrm{self}}$, by~\eqref{eq-total-length}, we have
\begin{align}\label{eq-b-up-1}
b \leq 	R_{\mathrm{self}}.
\end{align}
Note that  $\ell \geq T / n \geq 40\log n \geq 20$. If $R_{\mathrm{out}} = 0$, then $I_{\mathrm{self}}$ contains only a single interval $[2,\ell]$.
Thus $b = 0$ and we have $b \leq 2R_{\mathrm{out}}$.
Otherwise $R_{\mathrm{out}} \geq 1$. 
By property \ref{item:e-3} above, for each $j\in[a]$, 
it holds that either $r_j = \ell$ or $\ext{e}_{r_j+1} \in N_{\mathrm{out}}\tp{\ext{e}_{r_j}}$ (namely $r_j+1 \in S_{\mathrm{out}}$).
This implies $b \leq R + 1 = R_{\mathrm{out}} + 1\le 2 R_{\mathrm{out}}$,
because there are at most one $(l_j,r_j) \in I^{(1)}_{\mathrm{self}}$ satisfying $l_j = r_j = \ell$.
Hence, in both cases, we have
\begin{align}\label{eq-b-up-2}
b \leq 2R_{\mathrm{out}}.	
\end{align}
Combining~\eqref{eq-b-up-1} and~\eqref{eq-b-up-2} proves that $b \leq 	\min\{R_{\mathrm{self}},2R_{\mathrm{out}}\}$.

Finally, we construct the the subsequence 	$\ext{f}_1 < \ext{f}_2 <\ldots <  \ext{f}_m$ from IPP $\+P$.
We construct a subset $\+F$ by the following procedure.
\begin{itemize}
  \item For each $i \in S_{\mathrm{out}}$, we add $\ext{e}_i$ into $\+F$.
  \item For each interval $[l,r] \in I_{\mathrm{self}}$, for all integers $j \in [l,r]$ such that $(j-l)$ is an odd number, we add $\ext{e}_j$ into $\+F$.
    Note that by property \ref{item:e-2}, if $l > 2$, $\ext{e}_{l-1}$ is always in $\+F$ because of the previous rule.
  \item To finish, we sort all extended hyperedges in $\+F$ to obtain $\ext{f}_1 < \ext{f}_2 <\ldots <  \ext{f}_m$.
\end{itemize}
We now verify the three properties in \Cref{lemma-ind-path}.
\begin{itemize}
\item By the definition of IPP, for any $1\leq i,j \leq \ell$ satisfying $\abs{i-j} \geq 2$, $\ext{e}_i \cap \ext{e}_j = \emptyset$. Since $\ext{f}_1 < \ext{f}_2 <\ldots <  \ext{f}_m$ is a subsequence of $\+P$, the first property holds.
\item Fix an index $2 \leq j \leq m$. Suppose $\ext{f}_{j-1}=\ext{e}_{j_1}$ and $\ext{f}_{j} = \ext{e}_{j_2}$. If $\abs{j_1-j_2} \geq 2$, then $\abs{\ext{f}_i \cap \ext{f}_{i-1}} =0$. 
  Assume $j_1 + 1 = j_2$, which means that $\ext{e}_{j_1}$ and $\ext{e}_{j_2}$ are neighbours in extended hypergraph. 
  If $\ext{e}_{j_2} \in N_{\mathsf{out}}(\ext{e}_{j_1})$, by \Cref{observation-neighbor-2}, it holds that $\abs{\ext{f}_i \cap \ext{f}_{i-1}} = 1$. 
  Otherwise, $\ext{e}_{j_2} \in N_{\mathsf{self}}(\ext{e}_{j_1})$.
  There must exist an interval $[l,r] \in I_{\mathrm{self}}$ such that
  either $j_1,j_2\in[l,r]$ or $j_1\not\in [l,r]$ but $j_2\in[l,r]$.
  The first case is impossible because we do not add two consecutive indices in any interval of $I_{\mathrm{self}}$.
  The second case is also impossible because it implies $j_1=l-1$ and $j_2=l$, but $l$ cannot be added.
\item All extendeds hyperedge in $S_{\mathrm{out}}$ are added into $\+F$. For each interval $[l,r] \in I_{\mathrm{self}}$, $\lfloor\frac{r-l+1}{2}\rfloor$ extended hyperedges in $[l,r]$ are added into $\+F$. Hence, if $l \neq r$, the number of vertices in $[l,r]$ added to $\+F$ is at least $(r-l+1)/3$ (with $r=l+2$ being the worst case). By~\eqref{eq-total-length}, we have $m \geq R_{\mathrm{out}} + \frac{1}{3}(R_{\mathrm{self}} - b)$.
%
\end{itemize}
Hence, the subsequence $\ext{f}_1 < \ext{f}_2 <\ldots <  \ext{f}_m$ satisfies all the properties in \Cref{lemma-ind-path}.
\end{proof}

Now we are ready to prove \Cref{lemma-single-path}.

\begin{proof}[Proof of \Cref{lemma-single-path}]
  Let $\ext{f}_1 < \ext{f}_2 <\ldots <  \ext{f}_m$ be the subsequence given in \Cref{lemma-ind-path}.
For each $\ext{f}_i$ and $c \in [s]$, define a bad event $\+B_i(c)$ that for all $(j,w) \in \ext{f}_i$, either $X_j(w) \neq Y_j(w)$ or $X_j(w) = Y_j(w) = c$. Note that $\ext{f}_1 < \ext{f}_2 <\ldots <  \ext{f}_m$ is a subsequence in IPP $\+P$, the probability of $\+B(\+P)$ can be bounded as follows
\begin{align*}
\Pr_{\+C}\sqtp{\+B(\+P)} \leq \Pr_{\+C}\sqtp{\forall i \in [m], \exists c_i \in [s] \text{ s.t. } \+B_i(c_i)}.
\end{align*}
By~\eqref{eq-def-new-T}, it holds that $\ell \geq T / n \geq 40\log n \geq 20$.
By the last property in \Cref{lemma-ind-path}, $m \geq \frac{1}{3}(R_{\mathrm{out}} + R_{\mathrm{self}}) = \frac{\ell - 1}{3} > 6$.
We further truncate the last element $\ext{f}_m$ and obtain the following inequality 
\begin{align}
\label{eq-upper-B-P}
\Pr_{\+C}\sqtp{\+B(\+P)} \leq \Pr_{\+C}\sqtp{\forall i \in [m - 1], \exists c_i \in [s] \text{ s.t. } \+B_i(c_i)}\leq 
\sum_{\*c \in [s]^{m-1}}\Pr_{\+C}\sqtp{ \bigwedge_{i=1}^{m-1} \+B_i(c_i)},
\end{align}
where the second inequality follows from the union bound, and $\*c=(c_1,\dots,c_{m-1})\in[s]^{m-1}$.
The truncation ensures that all elements $(j,w) \in \cup_{i=1}^{m-1}\ext{f}_i$ satisfy $j > 0$. (See \Cref{definition-IPP} of IPPs.)

Fix $\*c \in [s]^{m-1}$, we bound the probability of the event $ \bigwedge_{i=1}^{m-1} \+B_i(c_i)$.
For each $1\leq i < m$, we define 
\begin{align*}
\ext{S}_i \defeq \begin{cases}
\ext{f}_i &\text{if } i = 1;\\
 \ext{f}_i \setminus \ext{f}_{i-1}	&\text{if } i > 1.	
 \end{cases}
\end{align*}
Since $\ext{S}_i \subseteq \ext{f}_i$, we have the following bound
\begin{align*}
  \Pr_{\+C}\sqtp{ \bigwedge_{i=1}^{m-1} \+B_i(c_i)} \leq 	\Pr_{\+C}\sqtp{ \bigwedge_{i=1}^{m-1}\tp{\forall (j,w) \in \ext{S}_i, \left(X_j(w) \neq Y_j(w)\right) \lor \left(X_j(w) = Y_j(w) = c_i\right) }}. 
\end{align*}
By the first property in \Cref{lemma-ind-path}, all $\ext{S}_i$ are mutually disjoint. 
Now we list all the extended vertices $\cup_{i=1}^{m-1}\ext{S}_i$ as $(j_1,w_1),(j_2,w_2),\ldots,(j_M,w_M)$, where $0 < j_1 < j_2 < \ldots <j_M$. 
For each $1 \leq  p \leq M$, there is a unique $i$ such that $(j_p,w_p)\in \ext{S}_i$ and we denote $\mathsf{idx}(j_p)\defeq i$. 
We define a bad event $\+A(p)$ that either $X_{j_p}(w_p) \neq Y_{j_P}(w_p)$ or  $X_{j_p}(w_p) = Y_{j_P}(w_p) = c_{\mathsf{idx}(j_p)}$.
Using the chain rule for the RHS of the inequality above, it holds that 
\begin{align*}
\Pr_{\+C}\sqtp{ \bigwedge_{i=1}^{m-1} \+B_i(c_i)} \leq \prod_{p = 1}^M \Pr_{\+C}\sqtp{ \+A(p) \mid \bigwedge_{p' < p}\+A(p')}.	
\end{align*}
Consider the probability of $\+A(p)$ conditional on all $\+A(p')$ for $p' < p$.
To simplify the notation, let $j = j_p > 0$ and $w = w_{p}$.
In the $j$-th update, $X_{j}(w)$ is sampled from the distribution $\nu_w^{X_{j-1}(V \setminus \{w\} )}$ and $Y_{j}(w)$ is sampled from the distribution $\nu_w^{Y_{j-1}(V \setminus \{w\} )}$.
For any $\tau \in [s]^{V \setminus \{w\}}$, it holds that 
\begin{align*}
\forall x \in [s],\quad \nu^\tau_w(x) = \sum_{y \in h^{-1}(x)}\mu^\tau_w(y). 
\end{align*}
Note that $\mu^\tau$ is actually the uniform distribution over a list colouring instance on $H$ where for each $u \neq w$, the colour list is $h^{-1}(\tau_u)$, and the colour list for $w$ is $[q]$. Hence, for each $u \neq w$,
the size of colour list of $u$ is at least $\lfloor q/s \rfloor$, and the size of colour list of $w$ is $q$, where $s = \ctp{\sqrt{q}}$.
Note that $q \geq 40\Delta^{\frac{2}{k-4}}$ and $k \geq  20$ implies $\lfloor q/s \rfloor^k \geq 2\mathrm{e}q^2k\Delta$.
By \Cref{lemma-local-uniform-colour}, for all $\tau \in [s]^{V \setminus \{w\}}$, it holds that
\begin{align*}
\forall y \in [q], \quad \frac{1}{q}\tp{1-\frac{4}{kq}} \leq \frac{1}{q}\exp\tp{-\frac{2}{kq}} \leq \mu^\tau_w(y)	\leq \frac{1}{q}\exp\tp{\frac{2}{kq}} \leq \frac{1}{q}\tp{1 + \frac{4}{kq}}.
\end{align*}
Hence, for any $\tau \in [s]^{V \setminus \{w\} }$, it holds that for any $x \in [s]$,
\begin{align*}
\frac{\abs{h^{-1}(x)}}{q}\tp{1 - \frac{4}{kq}} \leq  \nu^\tau_w(x) \leq  \frac{\abs{h^{-1}(x)}}{q}\tp{1 + \frac{4}{kq}}.
\end{align*}
Note that all the events $\+A(p')$ for $p'<p$ are determined by the updates from time 1 to time $j - 1$.
The above bounds for $\nu^\tau_w(x)$ holds for any configuration $\tau \in [s]^{V \setminus \{w\} }$.
In the $j$-th update step,
since $X_j(w)$ and $Y_j(w)$ are coupled by the optimal coupling and $\abs{h^{-1}(x)}\leq \ctp{q/s}$, we have the probability of $X_j(w) \neq Y_j(w)$ is at most $\frac{1}{2}\sum_{x \in [s]} \frac{\abs{h^{-1}(x)}}{q} \cdot \frac{8}{kq} = \frac{4}{kq}$, and the probability of $X_j(w) = Y_j(w) = c_i$ is at most $\frac{\ctp{q/s}}{q}\tp{1 + \frac{4}{kq}}$. Hence,
\begin{align*}
\Pr_{\+C}\sqtp{ \+A(p) \mid \bigwedge_{p' < p}\+A(p')}
&\le \frac{4}{kq} + \frac{\ctp{q/s}}{q}\tp{1 + \frac{4}{kq}} \overset{(\star)}{\leq} \frac{\ctp{q/s}}{q}\tp{1 + \frac{5}{k}}\\
&\leq \frac{1.16}{\sqrt{q}}\tp{1 + \frac{5}{k}}.
\end{align*}
where $(\star)$ holds because $\frac{\ctp{q/s}}{kq} \geq \frac{4}{kq}$ if $q \geq 40$ and the last inequality is due to $\ctp{q/s}\le 1.16 \sqrt{q}$. This implies
\begin{align*}
\Pr_{\+C}\sqtp{ \bigwedge_{i=1}^{m-1} \+B_i(c_i)} \leq \prod_{p = 1}^M \tp{\frac{1.16}{\sqrt{q}}\tp{1 + \frac{5}{k}}} = \prod_{i=1}^{m-1}\tp{\frac{1.16}{\sqrt{q}}\tp{1 + \frac{5}{k}}}^{\abs{\ext{S}_{i}}}.		
\end{align*}
By the second property in  \Cref{lemma-ind-path} and the definition $\ext{S}_i$, it holds that 
\begin{align*}
\forall 1 \leq i \leq m, \quad \abs{\ext{S}_i}	\geq k -1.
\end{align*}
Combining with~\eqref{eq-upper-B-P}, we have
\begin{align*}
\Pr_{\+C}\sqtp{\+B(\+P)} &\leq \sum_{\*c \in [s]^{m-1}}\Pr_{\+C}\sqtp{ \bigwedge_{i=1}^{m-1} \+B_i(c_i)} \leq \sum_{\*c \in [s]^{m-1}}\tp{\frac{1.16}{\sqrt{q}}\tp{1 + \frac{5}{k}}}^{(m-1)(k-1)}\\
&\leq \tp{ s \tp{ \frac{1.16}{\sqrt{q}}\tp{1 + \frac{5}{k}}}^{k-1} }^{m-1}. 	
\end{align*}
Now we claim that
\begin{align*}
  s \tp{ \frac{1.16}{\sqrt{q}}\tp{1 + \frac{5}{k}}}^{k-1} \le \frac{1}{10^3\Delta k^6}.
\end{align*}
Using $s=\lceil\sqrt{q}\rceil\le 1.16\sqrt{q}$,
it suffices to show that
\begin{align*}
  1.16\times 10^3(1.16)^{k-1}\tp{1 + \frac{5}{k}}^{k-1} \Delta k^6 \le q^{(k-2)/2}.
\end{align*}
Using $\tp{1 + \frac{5}{k}}^{\frac{2(k-1)}{k-2}}\le 1.7$ and $k^{12/(k-2)}\le 7.4$ for $k\ge 20$,
we further simplifies the condition into
\begin{align*}
  q \ge 7.4\times 1.7 \times (1.16\times 10^3)^{2/(k-2)} (1.16)^{2(k-1)/(k-2)} \Delta^{2/(k-2)},
\end{align*}
which is implied by $q \geq 40 \Delta^{\frac{2}{k-4}}$ and $k \geq 20$.

The claim implies that 
\begin{align*}
\Pr_{\+C}\sqtp{\+B(\+P)} \leq \tp{\frac{1}{10^3\Delta k^6}}^{m-1} = 10^3\Delta k^6\tp{\frac{1}{10^3 \Delta k^6}}^{m}.	
\end{align*}
Finally, by the third property in \Cref{lemma-ind-path}, we have
\begin{align*}
\Pr_{\+C}\sqtp{\+B(\+P)} \leq 10^3 \Delta k^6\tp{\frac{1}{10^3 \Delta k^6}}^{R_{\mathrm{out}} + \frac{1}{3}\tp{R_{\mathrm{self}} - b}}.		&\qedhere
\end{align*}
\end{proof}

\bibliographystyle{alpha}
\bibliography{Linear-LLL}

\end{document}